\documentclass[reqno]{imsart}
\setattribute{journal}{name}{}

\renewcommand\baselinestretch{1.}

\RequirePackage[OT1]{fontenc}
\RequirePackage{amsthm,amsmath}
\RequirePackage{natbib}
\RequirePackage[colorlinks,citecolor=blue,urlcolor=blue,breaklinks=true]{hyperref}
\RequirePackage{hypernat}
\RequirePackage{comment}
\RequirePackage{graphicx,subfigure,latexsym,amssymb}
\RequirePackage{float,epsfig,multirow,rotating}
\RequirePackage{upgreek,wrapfig}
\RequirePackage{listings}
\RequirePackage{pzccal, breakcites}

\startlocaldefs
\newcommand {\ctn}{\citet} 
\newcommand {\ctp}{\citep}       

\numberwithin{equation}{section}
\theoremstyle{plain}
\newtheorem{theorem}{Theorem}[section]
\newtheorem{definition}{Definition}[section]
\newtheorem{remark}{Remark}[section]
\newtheorem{proposition}{Proposition}[section]

\newcommand{\bSigma}{\boldsymbol{\Sigma}}

\newcommand{\bPsi}{\boldsymbol{\Psi}}

\newcommand{\bmu}{\boldsymbol{\mu}}

\newcommand{\bOmega}{\boldsymbol{\Omega}}

\newcommand{\bV}{\boldsymbol{V}}

\newcommand{\bR}{\boldsymbol{R}}

\newcommand{\bx}{\boldsymbol{x}}
\newcommand{\bX}{\boldsymbol{X}}

\newcommand{\bZ}{\boldsymbol{Z}}
\newcommand{\bz}{\boldsymbol{z}}

\newcommand{\bL}{\boldsymbol{L}}

\newcommand {\bN}{\mathcal{N}}

\endlocaldefs

\begin{document}
\renewcommand\baselinestretch{1.}
\begin{frontmatter}
\title{Soft Random Graphs in Probabilistic Metric Spaces $\&$ Inter-graph Distance}
\runtitle{SRGGs in Probabilistic Spaces $\&$ Inter-graph Distance}

\begin{aug}
\author{
{\fnms{Kangrui} \snm{Wang}\thanksref{t2}\ead[label=e2]{kwang@turing.ac.uk}}
\and
{\fnms{Dalia}
  \snm{Chakrabarty}\thanksref{t1,m1}\ead[label=e1]{dalia.chakrabarty@brunel.ac.uk}}
},
\thankstext{t2}{Alan Turing Institute} 
\thankstext{t1}{Department of Mathematics,
  Brunel University London}

\runauthor{Wang $\&$ Chakrabarty}

\affiliation{Alan Turing Institute, Brunel University London}


\address{\thanksmark{m1} Corresponding author\\
Department of Mathematics\\
Brunel University London\\
Uxbridge\\
Middlesex UB8 3PH\\
U.K.\\
\printead*{e1}
}

\end{aug}
\begin{abstract} 
{
We present a new method for learning Soft Random Geometric Graphs (SRGGs),
drawn in probabilistic metric spaces, with the connection function of the
graph defined as the marginal posterior probability of an edge random
variable, given the correlation between the nodes connected by that edge. In
fact, this inter-node
correlation matrix is itself a random variable in our learning strategy, and we
learn this by identifying each node as a random
variable, measurements of which comprise a column of a given multivariate
dataset. We undertake inference with Metropolis with a 2-block update scheme. The SRGG is shown to be generated by a non-homogeneous
Poisson point process, the intensity of which is location-dependent. Given the
multivariate dataset, likelihood of the inter-column correlation matrix is attained
following achievement of a closed-form marginalisation over all inter-row
correlation matrices. Distance between a pair of graphical models learnt 
given the respective datasets, offers the absolute correlation between the
given datasets; such inter-graph distance computation is our ultimate objective, and is
achieved using a newly introduced metric that resembles an
uncertainty-normalised Hellinger distance between posterior probabilities of
the two learnt SRGGs, given the respective datasets. Two sets of 
empirical illustrations on real data are undertaken, and application to
simulated data is included to exemplify the effect of incorporating
measurement noise in the learning of a graphical model.
}
\end{abstract}

\end{frontmatter}


\renewcommand\baselinestretch{1.5}
{
\section{Introduction}
\label{sec:intro}
\noindent
Graphical models of complex multivariate datasets, manifest intuitive
illustrations of the correlation structures of the data, and are of interest
in different disciplines \ctp{whittaker, large_net, plos2007, carvalhowest,
  dipankar}. Much work is present in the Statistics literature on the
intrinsic correlation structure of a multivariate dataset that comprises
multiple measurements of a vector-valued observable, where it is common
practice to model the joint probability distribution of a set of such
observable values, as matrix-Normal \ctp{multigraph, west_2016, wangwest}.

In this paper, we present a method for learning the correlation structure of a
multivariate dataset, and its graphical model. The general method is advanced
for iterative Bayesian learning of the correlation matrix, at each update of
which, a soft random geometric graph (SRGG)
\citep{penrose,penrose_2016,Giles2016} of the data is updated, where any such
SRGG is drawn in probabilistic metric space, \citep{menger, sklar}, such that
its connection function is the location-dependent marginal posterior
probability of an edge, given the correlation between the nodes that straddle
this edge, and the chosen cutoff radius is a probability in the probabilistic
metric space that the SRGG is drawn in. In fact, such an SRGG is shown to be
underlined by a non-homogeneous Poisson process with a an intensity that is
dependent on the location, or more precisely, the nodes, and thereby on the
inter-nodal correlation. Thus, the point process that generates this SRGG is
compounded with the process that generates the matrix-valued correlation
variable. The full graphical model of the data is then defined using the
sequence of SRGGs generated across iterations of this learning scheme. In this
method, inference on uncertainties of both the correlation matrix and
graphical model is undertaken, and we can acknowledge measurement errors in
learning both random structures as well.

The learning of the graph and correlation structure are undertaken Bayesianly,
using Bayesian inference that is implemented via Markov Chain Monte Carlo
(MCMC) inference techniques; to be precise, a
Metropolis-with-2-block-update-based algorithm is implemented \ctp{Robert04},
to make inference on the correlation matrix given the data, and on the SRGG
given the updated correlation.  However, for learning large networks, for
which such iterative inference is not practical, we modulate the method to
accommodate this concern. Then we undertake the network learning as a single
SRGG.

The ultimate aim behind our learning of the graphical model of a given data,
is to compute the distance between the graphical models learnt for a pair of
given datasets, in order to thereby compute the strength of the inter-data
correlation. We compute the distance between the graphical models learnt for
the respective dataset, using a new metric $\delta$ that we introduce (in
Theorem~\ref{theo:distmet}), where this metric is akin to the Hellinger
distance between the posterior probabilities of the graphical models given the
respective correlation structure of the given datasets, normalised by the
uncertainties in each of the learnt graphical models. Such distance informs on
the absolute correlation between the pair of multivariate datasets, for which
the graphical models are learnt.

Objective and comprehensive uncertainties on the Bayesianly learnt graphical
model of given multivariate data, are sparsely available in the
literature. Such uncertainties can potentially be very useful in informing on
the range of models that describe the partial correlation structure of
the data at hand. \ctn{madiganraftery} discuss a method for computing model
uncertainties by averaging over a set of identified models, and they advance
ways for the computation of the posterior model probabilities, by taking
advantage of the graphical structure, for two classes of considered models,
namely, the recursive causal models \ctp{carlincausal} and the decomposable
log-linear models \ctp{goodman}. This method allows them to select the ``best
models'', while accounting for model uncertainty. Our method on the other
hand, provides a direct and well-defined way of learning uncertainties of the
graphical model of a given multivariate data. 

However, we wish to extend such learning to higher-dimensional
data, for example, to a dataset that is cuboidally-shaped, given that it
comprises multiple measurements of a matrix-valued observable.  \ctn{hoff,
  xu2012}; Wang $\&$ Chakrabarty ({\url{https://arxiv.org/abs/1803.04582}}), advance methods to
learn the correlation in high-dimensional data in general. For a
rectangularly-shaped multivariate dataset, the pioneering work by
\ctn{wangwest} allows for the learning of both the inter-row and
inter-column covariance matrices, and therefore, of two graphical
models. \ctn{multigraph} extend this approach to high-dimensional data.
However, a high-dimensional graph showing the correlation structure amongst
the multiple components of a general hypercuboidally-shaped dataset, is not
easy to visualise or interpret. Instead, in this paper, we treat the
high-dimensional data as built of correlated rectangularly-shaped data slices,
given each of which, the inter-column (partial) correlation structure and
graphical model are Bayesianly learnt, along with uncertainties, subsequent to
our closed-form marginalisation over all inter-row correlation matrices
(in Section~\ref{sec:corr}, unlike in the work of \ctn{wangwest}). By invoking the uncertainties learnt
in the graphical models, we advance a new inter-graph distance
metric (Section~\ref{sec:hell}), based on the Hellinger distance \ctp{matu,juq} between the posterior probability densities of the pair of
graphical models that are learnt given the respective pair of such
rectangularly-shaped data slices. We use a corresponding affinity measure to
then infer on the correlation between the pair of datasets
(Section~\ref{sec:itisdist}), permitting the correlation structure of the
high-dimensional dataset thereby. Thus, by computing the
pairwise inter-graph distance between each
learnt pair of graphs, we can avoid the inadequacy of trying to capture
spatial correlations amongst sets of multivariate observations, by ``computing
partial correlation coefficients and by specifying and fitting more complex
graphical models'', as was noted by \ctn{guiness}. An additional advantage is
that our method offers
the inter-graph distance for two differently sized datasets. 

That we do not learn the graphical model as an Erdos-Renyi graph is because we
wish to utilise the fully Bayesian nature of the inference that we implement
here, without resorting to any unrealistic assumptions -- such as
decomposability. Effectively, we wish to learn the uncertainty-included
graphical model of noisy data, as distinguished from making inference on the
graph (i.e. writing its posterior) clique-by-clique. Also, we do not need to
be reliant on the closed-form nature of the posteriors to sample from,
i.e. we do not need to invoke conjugacy to affect our learning. Indeed, to
contextualise to a common practice in Bayesian learning of undirected graphs, a
Hyper-Inverse-Wishart prior is typically imposed on the covariance matrix of
the data, as this then allows for a Hyper-Inverse-Wishart posterior of the
covariance, which in turn implies that the marginal posterior of of any clique
is Inverse-Wishart -- a known, closed-form density \ctp{dawidlauritzen93,
  lauritzen96}.
Inference is then rendered easier, than when generating posterior samples from a
non-closed form posterior, (using techniques such as MCMC). Now, if the graph
is not decomposable, and a Hyper-Inverse-Wishart prior is placed on the
covariance matrix, the resulting Hyper-Inverse-Wishart joint posterior density
that can be factorised into a set of Inverse-Wishart densities, cannot be
identified as the clique marginals. Expressed differently, the clique
marginals are not closed-form when the graph is not decomposable. However,
this is not a worry in our learning as 
we can undertake our learning irrespective of the validity of
decomposability. Our pursued graphical model is a soft
RGG that is drawn in a (normed) probabilistic space, where
the location-dependent affinity measure between a pair of vertices of the
graph is conditional on the correlation learnt between the pair of random
variables at these two vertices.

This paper is organised as follows. The following section deliberates upon the
methodology development that we advance towards the learning of the SRGG and
the inter-column correlation matrix of a given dataset. In the subsequent
section, issues relevant to the inference on the unknowns is discussed, along
with definition of uncertainties in the learnt graphical model. The metric
used for computing the inter-graph distance, is then discussed in
Section~\ref{sec:hell}, while Section~\ref{sec:net} presents our modulated
learning methodology to accommodate challenges of learning large networks as
SRGGs. Section~\ref{sec:real} presents the empirical illustration on 2 real
datasets and comparison against existing results, while distance computation
between the learnt graphical models of these 2 data, is discussed in
Section~\ref{sec:real_hell}. In Section~\ref{sec:disease}, we learn the large
disease-phenotype network, and compare our results with those reported earlier
\ctp{hsg}. The paper is rounded up with a section that summarises the main
findings and the conclusions.

The attached Supplementary Material elaborates on certain aspects of our
work. This includes quantitative comparison of the results that we obtain
using the real datasets that we illustrate our methodology upon (Sections~7
and 9 of the Supplementary Material, along with outputs of such a comparative
exercise included in Figures~12, 13, 14 of the Supplement), with results that
are available in the literature, or obtained independently by us. Importantly,
detailed model checking is discussed in Section~5 of the Supplementary
Material, in the context of an empirical illustration made to a simulated
dataset (presented in Section~4 of the Supplement). Convergence signatures of
our MCMC chains are borne by the extra results that are included in Figures~9,
10, and 11 of the Supplement, in addition to discussions below in
Section~\ref{sec:real}.



\section{Background}
\label{sec:rgg}
\noindent
Let points $X_1,\ldots,X_p$ be independent, with the random variable $X_i\in{\cal X}\subseteq{\mathbb R}^d$ be s.t. $X_i\sim
f(\theta_1^{(i)},\ldots,\theta_q^{(i)})$, with the parameters of the
{\it{pdf}} of $X_i$ given as $\theta_1^{(i)},\ldots,\theta_q^{(i)}\in{\mathbb
  R}$. 

The $d$-dimensional, soft random geometric graph (SRGG) ${\cal
  G}_\phi(\bV)$ on the vertex set $\bV:=\{X_1,\ldots,X_p\}$, with
each of the $p$ points $X_1,\ldots,X_p\in{\cal X}\subseteq{\mathbb R}^d$
assigned a random coordinate in the box [0,1]$^d$, is 
s.t. probability of edge $G_{ij}$ between the $i$-th and $j$-th
nodes ($i\neq j; i,j\in\{1,\ldots,p\}$), is given by a function $\phi(\cdot)$
of the distance between point $X_i$ and point $X_j$. Here $\phi:{\cal X}\longrightarrow[0,1]$ is
referred to as the connection function, following \ctn{penrose_2016}.

\subsection{Probabilistic metric space: background}
\noindent
We draw our SRGG in a probabilistic metric space ${\cal X}$ \ctp{menger},
s.t. to any $X_1, X_2 \in{\cal X}$, a probability distribution
$F_{X_1,X_2}(X)$ is assigned, where $F_{X_1,X_2}(0)=0$ (similar to the
assignment of a non-negative number to any two points in a metric space). Let
all distributions that abide by the constraint $F_{\cdot,\cdot}(0)=0$, live in
space ${\cal F}_{+}\subset{\cal F}$, where probability distributions live in
${\cal F}$. We note that for $X_1=X_2$, the distribution
$F_{X_1,X_2}(\cdot)=\ell_0(\cdot)$, where the distribution $\ell_a(\cdot)\in{\cal F}_{+}$ is
s.t. $\ell_a(x) = 0 \text{ if } x \leq 0$ and $\ell_a(x) = 1 \text{ if } x >
0$, for $a\in{\mathbb R}_{\geq 0}$.

\begin{definition}
A probabilistic metric space is the triple
$$\{{\cal X}, F, \Delta\},$$
where
the probabilistic distance is $F_{X_1,X_2}:{\cal
  X}\times{\cal X}\longrightarrow{\cal F}_{+}$, $\forall X_1, X_2\in{\cal
  X}$, with $F_{X_1,X_2}(\cdot)=\ell_{0}(\cdot) \iff X_1=X_2$; \\
$F_{X_1,X_2}(\cdot) = F_{X_2,X_1}(\cdot)$, $\forall X_1, X_2\in{\cal
  X}$; and\\
triangle function $\Delta$ is s.t. $F_{X_1,X_3}(\cdot) \geq
\Delta(F_{X_1,X_2}, F_{X_2,X_3})(\cdot)$, $\forall X_1, X_2, X_3\in{\cal
  X}$,\\
with the triangle function $\Delta$ defined as a binary operation on ${\cal
  F}_+$, with respect to the triangle norm $D$ s.t. 
$$\Delta_D(F_{X_1,X_2}, F_{X_3,X_4})(x) = \sup[D(F_{X_1,X_2}(u),F_{X_3,X_4}(v)); u+v=x],$$
where the triangle norm $D$ \ctp{fodor} is defined as the binary operation on interval
[0,1], s.t. $D(y_1,y_2)=D(y_2,y_1)\forall y_1,y_2 \in [0,1]$; $D(y_1,y_2) \geq
D(y_3,y_4)\forall y_1\geq y_3; y_2\geq y_4; y_1,y_2,y_3,y_4 \in [0,1]$;
$D(y_1,D(y_2,y_3))=D(D(y_1,y_2),y_3)$ \\$\forall y_1,y_2,y_3\in[0,1]$; and $D(y,1)=y, \forall y \in [0,1]$, 
\end{definition}

\subsection{Marginal posterior of edge parameter as connection function}
\noindent
We draw our SRGG in a probabilistic metric space ${\cal X}$.

\begin{definition}
Points $X_1,\ldots,X_p$ are assigned random coordinates in
a 2-dimensional square, to construct a 2-dimensional SRGG on the vertex set
$\bV=\{X_1,\ldots,X_p\}$, with the probability of the edge between $X_i$ and
$X_j$ (where $X_i\neq X_j; X_i,X_j\in\bV$) given by the marginal posterior
probability of the random edge variable $G_{ij}$, conditional on the (partial)
correlation $\rho_{ij}$ between the random variables $X_i$ and $X_i$,
i.e.the connection function in our SRGG is the marginal posterior of
$G_{ij}$, given $\rho_{ij}$: 
$$m(G_{ij}\vert \rho_{ij}).$$ 
\end{definition}

\begin{remark}
  In our SRGG, the connection function is the affinity measure between a pair
  of points, defined by the marginal posterior of the edge variable connecting
  them in a (normed) probabilistic metric space \ctp{menger}. This affinity
  measure is complementary to the distance between this pair of points in this
  space.
\end{remark}

An immediate definition of the posterior probability density of our SRGG
would be the joint posterior probability of
the edge parameters ($\{G_{ij}\}_{i\neq j; i,j=1}^p$) given the partial
correlation structure, as:
$$\pi(G_{11}, G_{12},\ldots G_{p\:p-1}\vert \bR) \propto \ell(G_{11}, G_{12},\ldots G_{p\:p-1}\vert \bR)\:\pi_0(G_{11}, G_{12},\ldots G_{p\:p-1}),$$
where $\pi_0(G_{11}, G_{12},\ldots G_{p\:p-1})$ is the prior probability
density on the edge parameters $\{G_{ij}\}_{i\neq j; i,j=1}^{p}$, and $\ell(G_{12},\ldots, G_{1p},G_{23},\ldots,G_{2p},G_{34},\ldots,G_{p\:p-1}\vert
\bR)$ is the likelihood of the edge parameters, given the partial correlation
matrix $\bR=[\rho_{ij}]$. 

We choose a prior on $G_{ij}$ that is $Bernoulli(0.5)$ $\forall i,j$,
i.e. $$\pi_0(G_{11}, G_{12},\ldots G_{p\:p-1})
=\displaystyle{\prod\limits_{i,j=1; i\neq j}^p 0.5^{G_{ij}} 0.5^{1-G_{ij}}
};$$ thus, the prior is independent of the edge parameters. In applications
marked by more information, we can resort to stronger priors. However, as we
will soon see, this posterior definition needs updating.

We choose to define this likelihood as a function of the (squared) Euclidean
distance between the ``observation'', (i.e. the absolute value of $\rho_{ij}$),
and the unknown parameter $G_{ij}$, with the squared distance normalised by a
squared scale length, i.e. the variance parameter $\upsilon_{ij}$, for all
relevant $ij$-pairs. Thus, the unknown parameters in the model are the edge
parameters $\{G_{ij}\}_{i\neq j; i,j=1}^{p}$ and variance parameters
$\{\upsilon_{ij}\}_{i\neq j; i,j=1}^{p}$. In light of these newly introduced
variance parameters, we rewrite our likelihood (of the unknown edge and
variance parameters, given the partial correlation matrix),as:
{{$$\ell\left(G_{12},\ldots,G_{1p},G_{23},\ldots,G_{2p},\ldots,G_{p\:p-1},\upsilon_{12},\ldots,\upsilon_{1p},\upsilon_{23},\ldots,\upsilon_{2p},\ldots,\upsilon_{p\:p-1}\vert\bR\vert\right).$$}}
Now, we choose a parametric form of the likelihood, given its anticipated
properties. Likelihood increases (decreases) as distance between $\vert\rho_{ij}\vert$ and
$G_{ij}$ decreases (increases). Also the likelihood is invariant to change of
sign of $\vert\rho_{ij}\vert-G_{ij}$. Given this, we model our likelihood of the
edge and variance parameters, given $\bR$ as
$$\ell\left(G_{12},\ldots,G_{1p},G_{23},\ldots,G_{2p},\ldots,G_{p\:p-1},\upsilon_{12},\ldots,\upsilon_{1p},\upsilon_{23},\ldots,\upsilon_{2p},\ldots,\upsilon_{p\:p-1}\vert\bR\right) $$
\begin{equation}
\displaystyle{
=\prod\limits_{i\neq j; i,j=1}^{p} \frac{1}{\sqrt{2\pi\upsilon_{ij}}}
\exp\left[-\frac{\left(G_{ij} - \vert \rho_{ij}\vert\right)^2}{2\upsilon_{ij}} 
\right]}, 
\label{eqn:folded}
\end{equation}
where the variance parameters $\{\upsilon_{ij}\}_{i\neq j;i,j=1}^p$ are
indeed hyper-parameters that are also learnt from the data. The variance
parameters are assigned uniform prior probability in the interval $[0,1]$. 

Thus, the joint posterior probability of the edge and variance parameters is 
\begin{equation}
\pi(\{G_{ij}\}_{i\neq j; i,j=1}^p, \{\upsilon_{ij}\}_{i\neq j; i,j=1}^p \vert
\bR) = 
K\displaystyle{
\prod\limits_{i\neq j; i,j=1}^{p} \frac{1}{\sqrt{2\pi\upsilon_{ij}}}
\exp\left[-\frac{\left(G_{ij} - \vert \rho_{ij}\vert\right)^2}{2\upsilon_{ij}} 
\right]}, 
\label{eqn:folded_post}
\end{equation}
where $K>0$ is a constant that incorporates information on the $ij$-independent
Bernoulli(0.5) prior, Uniform[0,1] prior, and probability of $\bR$ that
defines the denominator of Bayes rule.

Then the marginal posterior probability of the $ij$-th edge parameter
$G_{ij}$, given the partial correlation $\rho_{ij}$ between $X_i$ and $X_j$,
is: 
\begin{eqnarray}
m(G_{ij}\vert \rho_{ij}) &=& K \displaystyle{\int\limits_{0}^1 \frac{1}{\sqrt{2\pi\upsilon}}
\exp\left[-\frac{\left(G_{ij} - \vert \rho_{ij}\vert\right)^2}{2\upsilon} 
\right] d\upsilon} \nonumber \\
&=& K \displaystyle{\left[
    \sqrt{\frac{2}{\pi}}\exp\left(\frac{-\left(G_{ij} - \vert
          \rho_{ij}\vert\right)^2}{2}\right) - 
\left(\vert G_{ij} - \vert \rho_{ij}\vert\vert\right) 
{\mbox{erfc}}\left(\frac{\vert G_{ij} - \vert
    \rho_{ij}\vert\vert}{\sqrt{2}}\right)\right]},\nonumber \\
&&
\label{eqn:marg_final}
\end{eqnarray}
where ${\mbox{erfc}}(\cdot)$ is the complimentary error function. 

\begin{theorem}
Setting the global constant $K=1$, separation $D_{ij}$ between
points $X_i\in{\cal X}$ and $X_j\in{\cal X}$ is a norm in the probabilistic
metric space ${\cal X}$, where
\begin{eqnarray}
D_{ij} &:=& m(G_{ij}\vert \rho_{ij}) + \vert G_{ij} -
\vert\rho_{ij}\vert\vert,\\ \nonumber
&=& \displaystyle{\left[
    \sqrt{\frac{2}{\pi}}\exp\left(\frac{-\left(G_{ij} - \vert
          \rho_{ij}\vert\right)^2}{2}\right) + 
\left(\vert G_{ij} - \vert \rho_{ij}\vert\vert\right) 
{\mbox{erf}}\left(\frac{\vert G_{ij} - \vert \rho_{ij}\vert\vert}{\sqrt{2}}\right)\right]},\\ \nonumber
\label{eqn:dist_aff}
\end{eqnarray}
where we have recalled that the error function ${\mbox{erf}}(\cdot) = 1 - {\mbox{erfc}}(\cdot)$,
and $m(G_{ij}\vert \rho_{ij})$ is defined in Equation~\ref{eqn:marg_final}.
\end{theorem}
The proof of this theorem is in Section~1 of the attached Supplementary Materials.

The global scale $K$ is chosen to ensure that $D_{ij}$ is a metric in the
probabilistic metric space ${\cal X}$.

Figure~\ref{fig:margdist} shows the variation with $\vert
G_{ij}-\vert\rho_{ij}\vert\vert$, in the unscaled distance $D_{ij}$ between
the $i$-th and $j$-th nodes; the unscaled affinity measure
$m(G_{ij}\vert\rho_{ij})$ between points $X_i$ and $X_j$; and the difference
between the unscaled distance and the affinity measure, for $G_{ij}$ set to 1
(trends displayed in the left panel), and $G_{ij}$ set to 0 (results shown on
the right). In our SRGG, the $ij$-th edge exists with the probability that is
given by the probability that affinity $m(G_{ij}\vert\rho_{ij})$ between $X_i$
and $X_j$ exceeds cutoff probability $\tau$. Now, probability for the $ij$-th
edge to exist, increases with increase in $\vert\rho_{ij})\vert$; so we expect
$m(G_{ij}=1\vert\rho_{ij})$ to increase with $\vert\rho_{ij})\vert$. This
trend is borne in the results displayed in Figure~\ref{fig:margdist}.
From this figure, we also see that the affinity measure defined
in terms of the marginal posterior of the edge, given the partial correlation, is
complementary to the distance $D_{ij}$, in the sense that as affinity
increases, distance decreases, for a given value of $G_{ij}$, $\forall
i,j\in\{1,2,\ldots,p\}; i\neq j$.
\begin{figure}[!hbt]
{
\hspace*{0in}
\includegraphics[width=12cm]{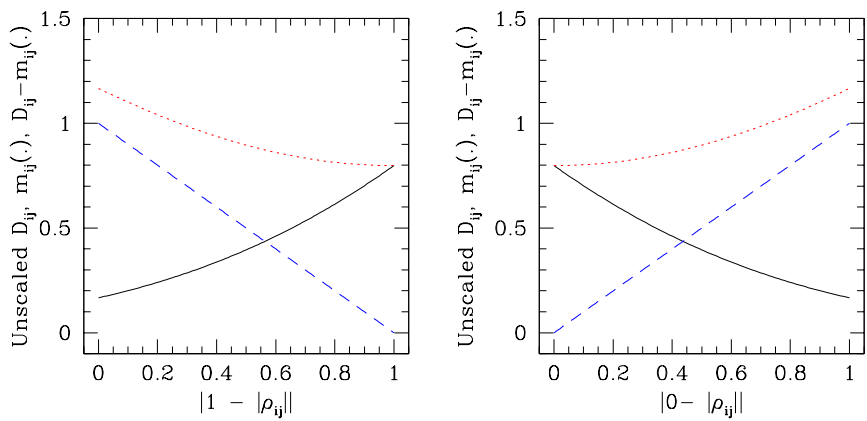}
}
\vspace*{-.0in}
\caption{Figure displaying variation with the input $\vert
G_{ij}-\vert\rho_{ij}\vert\vert$, of unscaled distance
$D_{ij}$ between the $i$-th and $j$-th nodes (in dotted lines); the unscaled affinity measure
$m(G_{ij}\vert\rho_{ij})$ between points $X_i$ and $X_j$ (in solid line); and the difference
between the unscaled distance and the affinity measure (in dashed lines), for $G_{ij}=1$ (left
panel) and $G_{ij}=0$ (right panel). Here the input
variable is the absolute of the
difference between value (0 or 1) of the edge variable $G_{ij}$ between $i$-th and
$j$-th nodes, and the absolute partial correlation $\vert\rho_{ij}\vert$
between points $X_i$ and $X_j$. As $\vert\rho_{ij}\vert$ increases,
probability for the edge connecting $X_i$ and $X_j$ to exist -- given by
the affinity measure $m_{ij}(G_{ij}=1\vert \rho_{ij})$ between points $X_i$
and $X_j$ in our SRGG to exceed a chosen cutoff probability -- increases, while the probability for this
edge to not exist (i.e. for $G_{ij}$ to be 0) decreases, as shown on the
right. Distance between points $X_i$ and $X_j$ is defined to be complimentary
to the affinity.}
\label{fig:margdist}
\end{figure}



\subsection{Edge set of our SRGG}
\label{sec:edgeset}
\begin{definition}
\label{defn:edgeset}
  Then edge $G_{ij}$ exists in the graph, independently of any other edge, if
  and only if, affinity between $X_i$ and $X_j$ exceeds a threshold
  probability $\tau$, i.e.
$$m(G_{ij}\vert \rho_{ij}) \geq \tau.$$
Thus, the value of the $ij$-th edge parameter $G_{ij}$ is
\begin{eqnarray}
g_{ij} = 1 && \quad {\mbox{if}} \quad m(G_{ij}\vert \rho_{ij}) \geq \tau; i\neq j; i,j\in\{1,\ldots,p\}\\ \nonumber
g_{ij} = 0 && \quad {\mbox{otherwise}}
\label{eqn:edge_defn}
\end{eqnarray}
Here the $ij$-th edge exists in the graph if $g_{ij}=1$, and does not exist if
$g_{ij}=0$. Thus, the edge set of our graph is
$$E_p = \{G_{ij}: m(G_{ij}\vert \rho_{ij}) \geq \tau; X_i\neq X_j; X_i,X_j\in\bV\}.$$ 
\end{definition}

\begin{remark}
As $\rho_{ij}$ is dependent on the $i$-th an
$j$-th vertices, marginal $m(G_{ij}|\rho_{ij})$ that is a function of
$\rho_{ij}$, is $i$ and $j$-dependent too, implying that the connection
function is location-dependent in this graph. Then following
\ctn{penrose_2016}, our notation for this
location-dependent SRGG defined at threshold parameter $\tau$ is ${\cal G}_{m,
  \bR}(\bV, \tau)$, where the partial correlation matrix is
$$\bR^{(p\times p)}=[\rho_{ij}],$$
where cardinality of $\bV$ is $p$.
\end{remark}

\subsection{Point process background for our SRGG}
\label{sec:point}
\noindent
In this sub-section, we present the underlining Point Process (PP) that can
generate our Soft RGG ${\cal G}_{m,\bR}(\bV, \tau)$. To undertake this, we now
model the points $X_i,X_j\in{\cal X}$, as random variables $X_i$ and $X_j$
that are Normally distributed with means $\mu_i$ and $\mu_j$, and the variance
$\sigma^2$.

\begin{theorem}
Distance between random variables $X_i, X_j\in{\cal X}$ -- where ${\cal X}$ is
the host probabilistic metric space -- is given as:
$$ D(X_i, X_j) = \displaystyle{\frac{2\sigma}{\sqrt{\pi}}
\exp\left(\frac{-\vert\mu_i - \mu_j\vert^2}{4\sigma^2}\right) +
\vert\mu_i - \mu_j\vert {\mbox{erf}}\left(\frac{\vert\mu_i -
    \mu_j\vert}{2\sigma}\right)}, $$
with $X_i\sim {\cal N}(\mu_i, \sigma^2)$, $X_j\sim {\cal N}(\mu_j, \sigma^2)$.
\end{theorem}

\begin{proof}
From definition of distance in probabilistic metric spaces, it follows that 
distance between $X_i$ and $X_j$ in ${\cal X}$ is
$$D(X_i, X_j) := \displaystyle{\int\limits_{-\infty}^\infty
  \int\limits_{-\infty}^\infty \vert x_i - x_j\vert f_{X_i}(x_i) f_{X_j}(x_j)
  dx_i dx_j}.$$
Then for $X_i\sim {\cal N}(\mu_i, \sigma^2)$, $X_j\sim {\cal N}(\mu_j,
\sigma^2)$, distance between these Normally distributed variables is
\begin{eqnarray}
&& D_{NN}(X_i, X_j) \nonumber \\
&:=& \displaystyle{\frac{1}{2\pi\sigma^2}
\int\limits_{-\infty}^\infty  \int\limits_{-\infty}^\infty 
\vert x_i - x_j\vert \exp\left(-\frac{(x_i - \mu_i)^2}{2\sigma^2}\right)
\exp\left(-\frac{(x_j - \mu_j)^2}{2\sigma^2}\right)   dx_i dx_j}\nonumber \\
&=& \displaystyle{
 \frac{2\sigma}{\sqrt{\pi}}\exp\left(-\frac{(\mu_i-\mu_j)^2}{4\sigma^2}\right) 
+ \vert \mu_i - \mu_j\vert {\mbox{erf}}\left(\frac{\vert\mu_i
    -\mu_j\vert}{2\sigma}\right)}, 
\label{eqn:d_NN}
\end{eqnarray}
where the subscript ``$NN$'' in the LHS qualifies the distance as between 2
Normally distributed r.v.s.
\end{proof}

\begin{proposition}
\label{prop:main}
SRGG ${\cal G}_{m,\bR}(\bV, \tau)$ drawn with affinity cut-off $\tau$, in the probabilistic metric space
${\cal X}$, with affinity measure $m(G_{ij}\vert \rho_{ij})$ defined as in
Equation~\ref{eqn:marg_final}, results from randomly placing Normally
distributed random variables with respective means, and a global variance
$\sigma^2$, in ${\cal X}$, as long as we set:
$$\sigma \equiv \displaystyle{\frac{1}{\sqrt{2}}}; \quad 
\vert\mu_i - \mu_j\vert\equiv \vert G_{ij} - \vert \rho_{ij}\vert\vert. $$
Here matrix $\bR=[\rho_{ij}]$.
\end{proposition}

\begin{proof}
We see in Equation~\ref{eqn:d_NN} that setting $2\sigma=\sqrt{2}$, and 
$\vert\mu_i - \mu_j\vert\equiv \vert G_{ij} - \vert \rho_{ij}\vert\vert$, 
implies distance between Normally distributed r.v.s $X_i\sim{\cal
  N}(\mu_i,\sigma)$, $X_j\sim{\cal N}(\mu_j,\sigma)$, with $X_i,X_j\in{\cal X}$
is
$$D_{NN}(X_i, X_j)=\displaystyle{
 \frac{\sqrt{2}}{\sqrt{\pi}}\exp\left(-\frac{(G_{ij}-\rho_{ij})^2}{2}\right) 
+ \vert G_{ij} - \rho_{ij}\vert {\mbox{erf}}\left(\frac{\vert G_{ij}
    -\vert\rho_{ij}\vert\vert}{\sqrt{2}}\right)} = D_{ij},$$ 
where $D_{ij}$ is defined in Equation~\ref{eqn:dist_aff} in terms of the
complimentary connection function, (i.e. complimentary affinity measure) of our SRGG, as
$D_{ij} := m(G_{ij}\rho_{ij})+  \vert G_{ij}-\vert\rho_{ij}\vert\vert$.
\end{proof}


\begin{theorem}
SRGG ${\cal G}_{m,\bR}(\bV, \tau)$, drawn with affinity cut-off $\tau$, in the probabilistic metric space
${\cal X}$, with affinity measure $m(G_{ij}\vert \rho_{ij})$ defined as in
Equation~\ref{eqn:marg_final}, is generated by a non-homogeneous Poisson point
process. 
\end{theorem}

\begin{proof}
Vertex set of SRGG ${\cal G}_{m,\bR}(\bV, \tau)$ is $\bV=\{X_1,\ldots,X_p\}$.\\
For any $i\in\{1,\ldots,p\}$, let ball $B(X_i, a)$ be drawn in ${\cal X}$,
centred at $X_i$, with radius $a$, where $a \in [0,1]$, given that radius of a
ball in ${\cal X}$ is a probability.\\
Let r.v. $N(a) := $ number of elements of $\bV$ inside $B(X_i, a)$.\\
Then $\forall B(X_i,a) \subset {\cal X}$, recalling that $G_{ij}=0$ if
$m(G_{ij}\vert\rho_{ij}) < \tau$, the expectation of $N(a)$ is:
\begin{eqnarray}
{\mathbb E}[N(a)] &=& \displaystyle{\sum\limits_{j=1}^p
  f_{X_i}(\mu_i,\sigma)\pi a^2 H(m(G_{ij}\vert\rho_{ij})-\tau) 
  },\\ \nonumber
&=& f_{X_i}(\mu_i,\sigma)\pi a^2 \displaystyle{\sum\limits_{j=1}^p
  H(m(G_{ij}\vert\rho_{ij}) - \tau) 
  }, \\ \nonumber
&\equiv& f_{X_i}(\mu_i,\sigma)\pi a^2 Q_{\bR,\tau}\nonumber
\label{eqn:heave}
\end{eqnarray}
where the Heaviside function $H(\cdot)$ is defined as
\begin{eqnarray}
H(x) &=& 1 \quad {\mbox{if}}\:\: x \geq 0 \\ \nonumber
&=& 0 \quad {\mbox{if}}\:\: x < 0 \\ \nonumber
\end{eqnarray}
and we introduce the notation:
$$ Q_{\bR,\tau}:=
\displaystyle{\sum\limits_{j=1}^p H(m(G_{ij}\vert\rho_{ij})-\tau)} $$
We further define 
$$\lambda_i := f_{X_i}(\mu_i,\sigma) Q_{\bR,\tau},$$
where we recall that
$$f_{X_i}(\mu_i,\sigma) = \displaystyle{\frac{1}{\sqrt{2\pi\sigma^2}} \exp\left(-\frac{(x_i - \mu_i)^2}{2\sigma^2}\right)}.$$
Then PP $\{X_1,\ldots,X_{i-1},X_{i+1},\ldots,X_p\}$ approximates a
non-homogeneous Poisson Process with intensity $\lambda_i$.\\
\end{proof}

\begin{remark}
Thus, our SRGG is underlined by a non-homogeneous Poisson
point process that has a location-dependent rate. Hence our SRGG is location-dependent. 
\end{remark}

\subsection{Dependence of graph on partial correlation matrix $\bR$}
\label{sec:corr}
\noindent
SRGG ${\cal G}_{m,\bR}(\bV, \tau)$ is learnt given the partial correlation
matrix $\bR$ that is itself learnt given the data. In fact, it is the
correlation matrix that is learnt given the data, and $\bR$ is computed at
every update of the correlation matrix. In this section we discuss the
learning of the correlation matrix given a multivariate dataset.

\begin{definition}
Let $\bX\in{\Xi}\subseteq{\mathbb R}^p$ be a $p$-dimensional observed vector,
with $\bX=(X_1,\ldots,X_p)^T$. \\
Let there be $n$ measurements of $X_j$, $j=1,\ldots,p$, so that the $n\times
p$-dimensional matrix ${\bf D}=[x_{ij}]_{i=1;j=1}^{n;p}$ is the data that
comprises $n$ measurements of the $p$-dimensional observable $\bX$. \\
Let the $i$-th realisation of $\bX$ be $\bx_i$, $i=1,\ldots,n$. \\
We model $\bX$, so that the set of $n$ realisations of this variable that comprises the data ${\bf D}$, is jointly matrix-Normal, i.e.
where this matrix-Normal density is parametrised by\\ 
---a mean matrix $\bmu^{(n\times p)}$;\\
---a covariance matrix $\bSigma_R^{(n\times p)}$, an element of which is the
covariance between a pair of rows in ${\bf D}$; \\
---a covariance matrix $\bSigma_C^{(p\times p)}$ that informs on inter-column
covariance in data ${\bf D}$. 
\end{definition}
We standardise the variable $X_j$ ($j=1,\ldots,p$) by its empirical mean
and standard deviation, into $Z_j$, s.t. the standardised version ${\bf D}_S$
of data ${\bf D}$ comprises $n$ measurements of the $p$-dimensional vector
$\bZ=(Z_1,\ldots,Z_p)^T$. Thus, 
$$z_{ij}=\displaystyle{\frac{x_{ij}-\bar{x}_j}{\Upsilon_j}},
\quad{\mbox{where}}\quad \bar{x}_j:=\displaystyle{\frac{\sum\limits_{i=1}^n
    x_{ij}}{n}}; 
\Upsilon_j^2 := \displaystyle{\frac{\sum\limits_{i=1}^n x_{ij}^2}{n}
  -\left(\frac{\sum\limits_{i=1}^n x_{ij}}{n}\right)^2}.$$ 
The $n\times p$-dimensional matrix ${\bf D}_S=[z_{ij}]$.

Then we model the joint probability of a set of 
$n$ measurements of $\bZ$, that comprises the
standardised data ${\bf D}_S$, to be matrix-Normal with zero-mean, i.e.
$$\{\bz_1,\ldots,\bz_n\}\sim{\cal MN}({\bf 0}, \bSigma_R^{(S)}, \bSigma_C^{(S)}),$$
i.e. the likelihood of the covariance matrices $\bSigma_R^{(S)}$ and $\bSigma_C^{(S)}$, given data ${\bf D}_S$, is matrix-Normal: 
$$\ell(\bSigma_R^{(S)}, \bSigma_C^{(S)}\vert {\bf D}_S) =$$
\begin{equation}
\displaystyle{
\frac{1}{(2\pi)^{\frac{np}{2}}
|\bSigma_C^{(S)}|^{\frac{p}{2}}
|\bSigma_R^{(S)}|^{\frac{n}{2}}}\times
\exp\left[
-\frac{1}{2}tr
\left\{(\bSigma_R^{(S)})^{-1} {\bf D}_S (\bSigma_C^{(S)})^{-1} ({\bf D}_S)^T\right\}
\right]
}.
\label{eqn:matrix_Normal_density}
\end{equation}
Here $\bSigma_R^{(S)}$ generates the covariance between the standardised variables $\bZ_i$ and $\bZ_{i^{/}}$, $i,i^{/}=1,\ldots,n$, (while 
$\bSigma_R$ generates the covariance between $\bX_i$ and $\bX_{i^{/}}$). Similarly, $\bSigma_C^{(S)}$ generates the correlation between columns of ${\bf D}_S$. 


\begin{theorem}
\label{theo:marg}
When the prior on $\bSigma_C^{(S)}$ is uniform,
the joint posterior probability density of the correlation matrices
$\bSigma_C^{(S)}, \bSigma_R^{(S)}$, given the standardised data ${\bf D}_S$ 
can be marginalised over the $n\times n$-dimensional inter-row correlation $\bSigma_R^{(S)}$, to yield
$$
\pi(\bSigma_C^{(S)}\vert{\bf D}_S)\propto
\displaystyle{\frac{1}{c\left(\bSigma_C^{(S)}\right){\Big{\vert}} \bSigma_C^{(S)}{\Big{\vert}}^{p/2} {\Big{\vert}}{\bf D}_S (\bSigma_C^{(S)})^{-1} ({\bf D}_S)^T {\Big{\vert}}^{\frac{n+1}{2}}}},
$$
where prior on $\bSigma_R^{(S)}$ is
the non-informative $\pi_0(\bSigma_R^{(S)})={\Big{\vert}}
\bSigma_R^{(S)}{\Big{\vert}}^{\alpha}$; $\alpha=\displaystyle{-\frac{n}{2}
  -1}$; and $\bSigma_C^{(S)}$ is assumed invertible. Here,
$c\left(\bSigma_C^{(S)}\right)$ is a function of $\bSigma_C^{(S)}$ that
normalises the likelihood. 

As posterior $\pi(\bSigma_C^{(S)}\vert{\bf D}_S)$ above is obtained for Uniform
prior on $\bSigma_C^{(S)}$, the likelihood of $\bSigma_C^{(S)}$ given data
${\bf D}_S$, i.e. {\it{pdf}} of ${\bf D}_S$ given $\bSigma_C^{(S)}$ is: 
$$
\ell(\bSigma_C^{(S)}\vert{\bf D}_S) \equiv f({\bf D}_S\vert\bSigma_C^{(S)})
\propto \pi(\bSigma_C^{(S)}\vert{\bf
  D}_S).
$$
\end{theorem}
The proof of this theorem is provided in Section~2 of the attached
Supplementary Materials.

\begin{proposition}
{
An estimator of the normalisation ${c}\left(\bSigma_C^{(S)}\right)$ of the posterior $\left[\bSigma_C^{(S)}\vert{\bf D}_S\right]$,
given 
by Theorem~\ref{theo:marg}
$${\hat c}\left(\bSigma_C^{(S)}\right) = 
\displaystyle{
{\mathbb E}_{Z^{/}_{n^{/}p}}
\left[\ldots\left[
{\mathbb E}_{Z^{/}_{11}}
\left[ 
\displaystyle{\frac{1} {{\Big{|}}\left( {\bf D}^{/} (\bSigma_C^{(S)})^{-1} 
({\bf D}^{/})^T\right){\Big{|}}^{\frac{n^{/}+1}{2}}}}  
\right]
\right]\ldots\right]
}.$$
}
\end{proposition}

\begin{proof}
{
We substitute the sequential computing of expectations
with respect to (w.r.t.) distribution of each element of dataset ${\bf D}^{/}$
-- as suggested in the statement of this Proposition -- with computation 
of the expectation w.r.t. the block ${\bf D}^{/}$ of these elements, where
${\bf D}^{/}$ abides by the inter-column correlation of
$\bSigma_C^{(S)}$. Thus, we approximate the normalisation
${c}^{/}\left(\bSigma_C^{(S)}\right)$ 
as:
$${\hat c}^{/}\left(\bSigma_C^{(S)}\right) = 
\displaystyle{{\mathbb E}_{{\bf D}^{/}_{S}}
\left[\displaystyle{
\frac{1}{{\Big{|}}
\left( {\bf D}^{/}(\bSigma_C^{(S)})^{-1} ({\bf D}^{/})^T\right)
{\Big{|}}^{\frac{n^{/}+1}{2}}}
} \right]
}.$$ 

We consider the sample of $k$ number of $n^{/}\times p$-dimensional data sets 
$\{ {\bf D}_1^{/}, \ldots, {\bf D}_k^{/} \}$, 
where ${\bf D}_q^{/}$ abides by inter-column correlation 
$\bSigma_C^{(S)}$ $\forall q=1,\ldots,k$,\\ 
s.t. ${\bf D}_q^{/} (\bSigma_C^{(S)})^{-1} ({\bf D}_q^{/})^T$ is positive definite
$\forall q=1,\ldots,k$, at each $t$. 

Then an estimator of ${\hat c}^{/}\left(\bSigma_C^{(S)}\right)$ is 
\begin{equation}
{\hat{c}}_C^{(S)} := \displaystyle{
\frac{1}{K}
{\sum\limits_{k=1}^K 
\frac{1}{{\Big{|}}\left( {\bf D}_k^{/} (\bSigma_C^{(S)})^{-1} 
({\bf D}_k^{/})^T\right)
{\Big{|}}^{\frac{n^{/}+1}{2}}
}
}
}.
\label{eqn:est}
\end{equation}
}
\end{proof}
Generation of a randomly sampled $n^{/}\times p$-sized data
set ${\bf D}_k^{/}$, with column correlation $\bSigma_C^{(S)}$, is undertaken.

\section{Bayesian Inference on SRGG and Correlation Matrix using MCMC, and the
Compounding of Processes Underlying the SRGG}
\noindent
We perform Bayesian inference on the matrix $\bR=[\rho_{ij}]$ of partial
correlations between each distinct pair of the observed variables in data
${\bf D}_S$, and simultaneously, on the $\bR$-dependent SRGG ${\cal G}_{m,\bR}(\bV,\tau)$ drawn
in probabilistic metric space ${\cal X}$ with affinity measure $m(G_{ij}\vert
\rho_{ij})$, on vertex set $\bV$, with cutoff probability $\tau$. The inference
that we undertake is Markov Chain Monte Carlo-based, i.e. MCMC-based. In
particular it is an implementation of the Metropolis-with-2-block-update.

\begin{remark}
In the implementation of the Metropolis-with-2-block-update, the
inter-column correlation matrix $\bSigma_C^{(C)}$ of the data is first updated
-- at which the updated partial correlation matrix $\bR$ is computed -- for the
SRGG to be then updated, at this updated $\bR$.
\end{remark}

\begin{remark}
\label{rem:sp}
{
Posterior probability $\pi\left(\bSigma_C^{(S)}\vert{\bf D}_S\right)$ of inter-column
correlation matrix $\bSigma_C^{(S)}$, given data ${\bf D}_S$, as given in Theorem~\ref{theo:marg},  
implies that the sequence of realisations of $\bSigma_C^{(S)}$, at successive iterations
of the MCMC chain:  $$\{\bSigma_C^{(S)_0}, \bSigma_C^{(S)_1},
\bSigma_C^{(S)_2},\ldots\},$$
is a continuous-valued, discrete time
stochastic process, with underlying probability density $\pi\left(\bSigma_C^{(S)}\vert{\bf D}_S\right)$.
}
\end{remark}

\begin{definition}
{
Given a learnt value of the inter-column correlation matrix
$\bSigma_C^{(S)}$, to compute the value $\rho_{ij}$ of the partial correlation
$R_{ij}$ between r.v.s $X_i$ and $X_j$, we first invert the $p\times
p$-dimensional matrix $\bSigma_C^{(S)}$ to yield the precision matrix: 
$$\bPsi:=\left(\bSigma_C^{(S)}\right)^{-1};\:\bPsi=[\psi_{ij}],$$ s.t. the
partial correlation matrix $\bR=[\rho_{ij}]$, where $\rho_{ij}$ is
\begin{equation}
\rho_{ij} = -\displaystyle{\frac{\psi_{ij}}{\sqrt{\psi_{ii} \psi_{jj}}}},\quad i\neq j,
\label{eqn:6}
\end{equation}
and $\rho_{ii}=1$ for $i=j$. 
}
\end{definition}

\begin{proposition}
{
Following Proposition~\ref{prop:main}, the non-homogeneous Poisson point
process with $\bR$-dependent intensities, that generates the SRGG ${\cal G}_{m,\bR}(\bV,\tau)$, is
compounded with the continuous-valued stochastic
process $\{\bSigma_C^{(S)_t}\}_{t\in\{0,1,2,\ldots\}}$ (discussed in
remark~\ref{rem:sp}), with underlying
density $\pi\left(\bSigma_C^{(S)}\vert{\bf D}_S\right)$ that is defined in Theorem~\ref{theo:marg}.  
}
\end{proposition}


\subsection{MCMC-driven definition of the 95$\%$ HPD credible regions on the
  learnt SRGG}
\label{sec:95}
\noindent
To
acknowledge uncertainties in the Bayesian learning of the sought SRGG, where
the uncertainties can be identified with a Bayesian 95$\%$ HPD credible region, 
we can suggest that the learnt SRGG only contains those edges, the affinity measures
of which exceed a global probability of 0.05. Then we are basically 
defining $\tau=0.05$ to define the edge set of the sought SRGG (see
Definition~\ref{defn:edgeset}). 

Within our MCMC-based inference, the random graph ${\cal
  G}_{m,\bR^{(t)}}(\bV,\tau)$, is sampled in the $t$-th
iteration of the inference scheme, given the inter-column correlation matrix
$\bSigma_C^{(S)_t}$ updated in this iteration, using which partial correlation
matrix $\bR^{(t)}=[\rho_{ij}^{(t)}]$ is updated in the $t$-th iteration;
$t=0,1,\ldots,N_{iter}$. 

In the $t$-th iteration, let affinity between r.v.s $X_i$ and $X_j$ ($X_i\neq
X_j; X_i,X_j\in\bV$) be given by the marginal $m(G_{ij}=g_{ij}^{(t)}\vert\rho_{ij}^{(t)})$.

Then the graphical model of data ${\bf D}_S$ is defined as the graph on
vertex set $\bV$, that includes edge
between $X_i$ and $X_j$, ($X_i\neq
X_j; X_i,X_j\in\bV$), if sample estimate of
${\mathbb E}[m(G_{ij}\vert\rho_{ij}^{(t)})]$ exceeds $\tau=0.05.$ We delineate this
definition formally in Definition~\ref{defn:95}.

A sample estimate of ${\mathbb E}[m(G_{ij}\vert\rho_{ij}^{(t)})]$ is
$${\hat{m}}(G_{ij}\vert\rho_{ij}^{(t)}) := {N_{ij}},$$
where $N_{ij}$ is the
the fraction of the post-burnin
($N_{post}$) number of iterations (where $N_{post}< N_{iter}+1$), in which the $ij$-th
edge exists, i.e. in which $G_{ij}$ takes the value 1,
$\forall\:i,j=1,2,\ldots,p,\:i\neq j$. Thus, variable
$N_{ij}$ takes the value  
\begin{equation}
n_{ij} := \displaystyle{\frac{\sum\limits_{t= N-N_{post}+1}^N g_{ij}^{(t)}}{N_{post}}},\quad i<j;\:i,j=1,\ldots,p.
\label{eqn:n}
\end{equation}

\begin{remark}
  $N_{ij}$ is the (post-burnin) sample mean of the affinity measure
  between $X_i$ and $X_j$ ($X_i,X_j\in\bV$), i.e. of the marginal of edge
  $G_{ij}$, conditional on the learnt inter-column (partial) correlation
  matrix of the given dataset. 

We define the $p\times p$ ``edge-parameter matrix'' as $$\bN:=[n_{ij}].$$
\end{remark}
Indeed, the $N_{ij}$ parameter is a function of the partial correlation $\rho_{ij}$, that is
itself learnt given this data, but for the sake of notational brevity, we do
not include this explicit $\bR$ dependence in our notation.

\begin{definition}
\label{defn:95}
The set of SRGGs:
$$\{{\cal G}_{m,\bR^{(t)}}(\bV,\tau)\}_{t=N-N_{post}+1}^N$$ 
sampled during the post-burnin part of any MCMC chain, on
the vertex set $\bV=\{X_1,\ldots,X_p\}$, for $\tau=0.05$, defines the
graphical model ${\hat{\cal G}}_{\bN,{\bf D}_S}(\bV,0.05)$
of the data ${\bf D}_S$, learnt within 95$\%$ HPD credible region, as the
graph on vertex set $\bV$, with the 
edge set:
$${\hat{E}}_p = \{G_{ij} : n_{ij} \geq 0.05; X_i\neq X_j; X_,X_j\in\bV\}.$$
Here, the (post-burnin) sample mean
of the affinity measure between $X_i$ and $X_j$ is $n_{ij}$
s.t. edge-parameter matrix is
$\bN=[n_{ij}]$.
\end{definition}

\subsection{Computational Details of Metropolis-with-2-block-update}
\label{sec:algo}
\noindent
Theorem~\ref{theo:marg} gives the posterior probability density of correlation matrix ${\bSigma}_C^{(S)}$, given
data ${\bf D}_S$. In our Metropolis-with-2-block-update based inference,
we update ${\bSigma}_C^{(S)}$--at which the partial correlation matrix $\bR$
is computed. Given this updated $\bR$, we then update the graph. 

In our learning of the $p\times p$-dimensional inter-column correlation
matrix ${\bSigma}_C^{(S)}=[s_{ij}]$ -- where $s_{ij}$ is the value of
r.v. $S_{ij}$ -- the $\displaystyle{\frac{p^2-p}{2}}$ non-diagonal
elements of the upper (or lower) triangle are learnt, i.e. the parameters
$S_{12}, S_{13},\ldots,$\\ $S_{1p}, S_{23},\ldots,S_{p-1\: p}$ are learnt. 

In the $t$-th iteration, $S_{ij}$ is proposed from a Truncated
Normal density that is left truncated at -1 and right truncated at 1, as
$$S_{ij}=s_{ij}^{(t*)}\sim{{\cal TN}}(s_{ij}^{(t*)}; s_{ij}^{(t-1)}, \sigma^2_{ij}, -1,
1),\quad,\forall\:i,j=1,\ldots,p;\:i\neq j,$$ where $\sigma_{ij}=\sigma_0^2\:\forall\:i,j$, is
the experimentally chosen variance, and the proposal mean is the current value
$s_{ij}^{(t-1)}$ of $S_{ij}$ at the end of the $t-1$-th iteration. 

At the 2nd
block of the $t$-th iteration, the SRGG is
updated, given the recently updated partial correlation matrix $\bR^{(t)}=[\rho_{ij}^{(t)}]$, s.t. the
proposed edge variable connecting the $i$-th to the $j$-th vertex is
$$ G_{ij}=g_{ij}^{(t\star)}\sim{Bernoulli}(g_{ij}^{(t\star)};\rho_{ij}^{(t)}).$$ 

This update also involves the likelihood of the
edge parameters $G_{ij}$ and variance parameters $\upsilon_{ij}$, $\forall
i\neq j; i,j\in\{1,\ldots,p\}$, given the updated partial correlation
matrix. We recall from Equation~\ref{eqn:folded} that this likelihood is
$$\ell\left(G_{12},\ldots,G_{1p},G_{23},\ldots,G_{2p},\ldots,G_{p\:p-1},\upsilon_{12},\ldots,\upsilon_{1p},\upsilon_{23},\ldots,\upsilon_{2p},\ldots,\upsilon_{p\:p-1}\vert\bR\right) $$
\begin{equation}
\displaystyle{
=\prod\limits_{i\neq j; i,j=1}^{p} \frac{1}{\sqrt{2\pi\upsilon_{ij}}}
\exp\left[-\frac{\left(G_{ij} - \vert \rho_{ij}\vert\right)^2}{2\upsilon_{ij}} 
\right]}. \nonumber
\end{equation}
The $ij$-th variance parameter is assigned a proposed value of 
$$\upsilon_{ij}^{(t\star)}\sim{\cal N}(\upsilon_{ij}^{(t\star)};
\upsilon_{ij}^{(t-1)}, w_{ij}^2),$$ 
where $w_{ij}^2$ is the experimentally chosen variance, and the mean is the
current value of the
$\upsilon_{ij}$ variable. 

As suggested in Theorem~\ref{theo:marg}, the correlation learning involves
computing 
$\left({\bSigma}_C^{(S)}\right)^{-1}$, $\vert {\bSigma}_C^{(S)}\vert$
and $\vert{\bf D}_S \left({\bSigma}_C^{(S)}\right)^{-1} \left({\bf
  D}_S\right)^T\vert$, in every iteration.
This calls for Cholesky decomposition of $\bSigma_C^{(S)}$ as
$\bL_C^{(S)}(\bL_C^{(S)})^T$, and of
${\bf D}_S \left({\bSigma}_C^{(S)}\right)^{-1} \left({\bf
  D}_S\right)^T$, into the (lower) triangular matrix $\bL$ and
$\bL^T$, while implementing ridge adjustment \ctp{wothke}. The latter computation follows the inversion of
$\bSigma_{C}^{(S)}$ into $(\bSigma_{C}^{(S)})^{-1}$, which is undertaken
using a forward substitution algorithm. (see Section~10 of the Supplementary
Materials).

Once the set of SRGGs sampled during the post-burnin part of the MCMC chain
are identified ($\{{\cal G}_{m,\bR^{(t)}}(\bV,\tau)\}_{t=N-N_{post}+1}^N$),
graphical model ${\hat{\cal G}}_{\bN,{\bf D}_S}(\bV,0.05)$ of the data ${\bf
  D}_S$ is then constructed, using this set of SRGGs.

\section{Inter-graph distance metric}
\label{sec:hell}
\noindent
We compute the Hellinger distance between the posterior probability of the
graphical model ${\hat{\cal G}}_{\bN_1,{\bf D}_1}(\bV,0.05)$ of the data ${\bf
  D}_1$, learnt within 95$\%$ HPD credible region, and the posterior
probability of the similarly learnt graphical model ${\hat{\cal
    G}}_{\bN_2,{\bf D}_2}(\bV,0.05)$ of the data ${\bf D}_2$, Distance between the
pair of uncertainty-included graphical models is computed using the Hellinger
metric, normalised by the uncertainty in the learning of each graphical model,
where such uncertainty 
is defined below (Definition~\ref{defn:uncert}). Here data ${\bf D}_1$ comprises $n_1$ rows and $p$ columns,
while data ${\bf D}_2$ comprises $n_2$ rows and $p$ columns. (As we soon
explain, we compute the Hellinger distance between the marginal posterior
probability of the edges in each of the considered pair of graphical models).

The inter-graph distance is computed, to inform on the absolute of the
correlation between the multivariate, disparately-sized datasets ${\bf D}_1$ and
${\bf D}_2$; in effect, the exercise can address the possible independence of
the {\it{pdf}}s that the two datasets are sampled from. This is of course a
hard question to address when the data comprise measurements of a
high-dimensional vector-valued observable.

\begin{definition}
{Square of Hellinger distance between two probability density functions $g(\cdot)$ and $h(\cdot)$ over a common domain ${\cal X}\in{\mathbb R}^m$, with respect to a chosen measure, is 
\begin{eqnarray}
D_H^2(g,f) &=& \displaystyle{\int\left(\sqrt{g(\bx)} - \sqrt{h(\bx)}\right)^2 d\bx} \nonumber \\ 
&=& \displaystyle{\int g(\bx)d\bx + \int h(\bx)d\bx -2\int \sqrt{g(\bx)}\sqrt{h(\bx)}d\bx } \nonumber \\ 
&=& \displaystyle{2\left(1 -\int \sqrt{g(\bx)}\sqrt{h(\bx)}d\bx\right) }.
\label{eqn:hell}
\end{eqnarray}}
\end{definition}
The Hellinger distance is closely related to the Bhattacharyya distance
\ctp{bhat} between two densities:
$D_B(g,f) =
\displaystyle{-log\left[\int\left(\sqrt{g(\bx)}\sqrt{h(\bx)}\right)^2
    d\bx\right]}$.

\begin{definition}
Consider the marginal posterior probability density of all the graph edge
parameters $\{G_{ij}\}_{i\neq j;i,j=1}^p$ given the partial correlation matrix $\bR_q$
(that is itself updated given the data ${\bf D}_S^{(q)}$); $q=1,2$.

In the $t$-th
iteration, value of the marginal posterior of all the edges\\
$\{G_{ij}^{(qt)}\}_{i\neq j;i,j=1}^p$, in the
$q$-th SRGG, given $\bR_{qt}$, is:
$$m(G^{(qt)}_{11},G^{(qt)}_{12},\ldots,G^{(qt)}_{p\;p-1}\vert \bR_{qt}),\quad
t=0,\ldots,N_{iter}.$$ 
Given the availability of the value of this marginal posterior density, only at
discretely sampled points in its support, (sampled at discrete times $t$), the
integral in the definition of the Hellinger distance is replaced by a sum in
our computation of the distance. 

Then for the $q$-th dataset, the marginal
posterior of all graph edge parameters in the $t$-th iteration is:
$$u_{q}^{(t)}:=m(G^{(qt)}_{11},G^{(qt)}_{12},\ldots,G^{(qt)}_{p\;p-1}\vert
\bR_{qt}),$$ 
which is employed to compute square of the (discretised version of the) Hellinger
distance between the two datasets as
\begin{equation}
D_H^2(u_{1},u_{2}) = \displaystyle{\frac{
\sum\limits_{t=N_{burnin}+1}^{N_{iter}}\left(\sqrt{u_{1}^{(t)}} - \sqrt{u_{2}^{(t)}}\right)^2}{N_{iter}-N_{burnin}}},
\label{eqn:sq_hell}
\end{equation}
\end{definition}
The Bhattacharyya distance can be similarly discretised.

However, MCMC does not provide normalised posterior probability densities --
we may employ Uniform (over identified finite intervals) priors on the variance parameters, the marginalised
posterior probability of the edge parameters is known only up to an unknown
scale. 

\begin{remark}
In the $t$-th iteration, MCMC provides value of 
logarithm $\ln(u_{q}^{(t)})$ of the un-normalised posterior of the edges of
the graph given the $q$-th data ($q=1,\: 2$). Hence the Hellinger
distance between the 2 datasets that we compute is only known
upto a constant normalisation $S$ that we use to scale both $u_{1}^{(t)}$ and
$u_{2}^{(t)}$, $\forall \; t=0,\ldots, N_{iter}$. 
\end{remark}

\begin{proposition}
Unknown normalisation $S$ that normalises $u_{1}^{(t)}$ and
$u_{2}^{(t)}$, is chosen to ensure that the scaled, log marginal of all graph edges in
the $t$-th iteration, is $\leq 0$, s.t.
$\exp\left(\frac{\ln(u_{m}^{(t)})}{s}\right)\in(0,1]$. 
Therefore we choose the global scale $S$ as: 
\begin{equation}
s:={\max}\{(\ln(u_{1}^{(0)}),
\ln(u_{1}^{(1)}),\ldots,\ln(u_{1}^{(N_{iter})}),
\ln(u_{2}^{(0)}),\ldots,\ln(u_{2}^{(N_{iter})})\}.
\label{eqn:es}
\end{equation} 
\end{proposition}

\begin{remark}
\label{rem:hell}
{
Squared Hellinger distance $D_H^2(u_{1},u_{2})$ between discretised posterior
probability densities of 2 graphical models, computed using
$\exp(\ln(u_q^{(t)})/s)$ in Equation~\ref{eqn:sq_hell}, is affected by 
scaling parameter $S$. This scale dependence is mitigated in our definition of 
the distance between 2 graphical models
as the difference between the ratio of this computed 
${D_H(u_{1},u_{2})}$, to the scaled uncertainty inherent in one graphical
model, and the ratio of ${D_H(u_{1},u_{2})}$, to the scaled uncertainty
in the other learnt graphical model. Such scaled uncertainty in a learnt
graphical model is defined in Definition~\ref{defn:uncert}.}
\end{remark} 

\begin{definition}
\label{defn:uncert}
The scaled (by a scale parameter $S=s$) uncertainty in learnt graphical model
${\hat{\cal G}}_{\bN_q,{\bf D}_q}(\bV,0.05)$ of data set ${\bf D}_q$, with
edge-parameter matrix $\bN_q$, is defined as 
\begin{eqnarray}
D_{max, s}(q) & := & {\max}\{\exp(\ln(u_{q}^{(0)})/s),
\exp(\ln(u_{q}^{(1)})/s),\ldots,\exp(\ln(u_q^{(N_{iter})})/s)\}-\nonumber \\
&&{\min}\{\exp(\ln(u_{q}^{(0)})/s),\exp(\ln(u_{q}^{(1)})/s),\ldots,\exp(\ln(u_q^{(N_{iter})})/s)\},\nonumber \\
\label{eqn:graph_new}
\end{eqnarray}
Thus, $D_{max,s}(q)$ provides separation between the
maximal and minimal (scaled values of) posteriors of graphs, generated in the
MCMC chain run using the $q$-th dataset. Therefore $D_{max,s}(q)$ defines
uncertainty of the graphical model learnt for this dataset.
\end{definition}

\begin{definition}
\label{defn:1}
{
For edge-parameter matrix $\bN_q$, for dataset
${\bf D}_q$, $q=1,2$,
separation between the two corresponding graphical models on vertex set $\bV$, 
learnt with
uncertainty defined as in Definition~\ref{defn:uncert} is
\begin{eqnarray}
&&\delta(\displaystyle{{\hat{\cal G}}_{\bN_1,{\bf D}_1}(\bV,0.05)},
\displaystyle{{\hat{\cal G}}_{\bN_2,{\bf D}_2}(\bV,0.05)}) \nonumber\\
&:= & {\Big{\vert}}\sqrt{D_H^2(u_{1},u_{2})}/D_{max,s}(1)-\sqrt{D_H^2(u_{1},u_{2})}/D_{max,s}(2){\Big{\vert}}\nonumber\\
&= & D_H(u_{1},u_{2})\displaystyle{{\Big{\vert}}\frac{1}{D_{max,s}(1)}
  - \frac{1}{D_{max,s}(2)}{\Big{\vert}}},
\label{eqn:delta12}
\end{eqnarray}
 where
the Hellinger distance $D_H(u_{1},u_{2})$, between the 2 graphical models, is
defined in Equation~\ref{eqn:sq_hell} and $D_{max, s}(q)$ is the uncertainty
in the the graphical model for data ${\bf D}_q$, as defined in
Equation~\ref{eqn:graph_new}, computed at the chosen value $s$ of scale $S$ (defined in
Equation~\ref{eqn:es}).}
\end{definition}


Alternatively, we could define a (discretised version of the) odds ratio of unscaled logarithm of the unnormalised posterior densities of the graphical models learnt using MCMC, given the two datasets, as $\displaystyle{\int\left(\log(g(\bx))-{\log(h(\bx))}\right)d\bx }$; such is then a divergence measure that we define as
\begin{equation}
O_\pi(u_{1},u_{2}):= \displaystyle{
\sum\limits_{t=N_{burnin}+1}^{N_{iter}}
\left[{\log(u_{1}^{(t)}) -\log({u_{2}^{(t)}})}\right]
}.
\label{eqn:div}
\end{equation}

\subsection{Suggested inter-graph separation $\delta(\cdot,\cdot)$, is an inter-graph distance}
\label{sec:itisdist}
\begin{theorem}
\label{theo:distmet}
{
Let $\delta(\displaystyle{{\hat{\cal G}}_{\bN_1,{\bf D}_1}(\bV,0.05)},
\displaystyle{{\hat{\cal G}}_{\bN_2,{\bf D}_2}(\bV,0.05)})$ be the separation 
defined as in Equation~\ref{eqn:delta12}, between 2 uncertainty-included graphical models, defined over vertex set
$\bV$, learnt for datasets ${\bf D}_1$ and ${\bf D}_2$. Here the graphical model
$\displaystyle{{\hat{\cal G}}_{\bN_q,{\bf D}_q}(\bV,0.05)}$ is an element of space
${\bOmega}$, $q=1,2$.  

Then our definition of this inter-graph separation $\delta:{\bOmega}\times {\bOmega}\longrightarrow{\mathbb
  R}_{\geq 0}$, is a distance function, or a metric.}
\end{theorem} 
The proof of this theorem is provided in Section~3 of the attached
Supplementary Materials.

\subsection{Absolute correlation between 2 multivariate datasets, from distance
between their graphical models}
\noindent
In this section, we introduce a model for the absolute correlation between 2
multivariate datasets, for which the uncertainty-included graphical models are
learnt, allowing for the inter-graph distance $\delta(\cdot,\cdot)$ to be
computed.

\begin{proposition}
  \label{prop:2} { For a given value of the inter-graph distance
    $\delta(\displaystyle{{\cal G}_{u,1}}, \displaystyle{{\cal
        G}_{u,2}})\in[0,\infty)$, between 2 learnt graphical models
    $\displaystyle{{\cal G}_{u,2}}\displaystyle{{\cal
        G}_{u,1}}\in{\bOmega}$, defined over vertex set $\{1,\ldots,p\}$,
    where the graphical model $\displaystyle{{\cal G}_{u,\cdot}}$ is learnt given data
    ${\bf D}_\cdot$, a 
    model for the absolute value of the correlation $\vert corr(\bZ_1,\bZ_2)\vert $ between the $p$-dimensional vector-valued observable
    $\bZ_1$, ($n_1$ measurements of which comprise dataset indexed by 1), and the
    $p$-dimensional observable $\bZ_2$, ($n_2$ measurements of which comprise
    dataset indexed by 2), is
$$\delta(\displaystyle{{\cal G}_{u,1}}, \displaystyle{{\cal
        G}_{u,2}}) = -\log\left(\vert corr(\bZ_1,\bZ_2)\vert\right),$$   
$$s.t.\:\: \vert corr(\bZ_1,\bZ_2)\vert = \exp[-\delta(\displaystyle{{\cal G}_{u,1}},
\displaystyle{{\cal G}_{u,2}})]\in(0,1].$$ 
}
\end{proposition}

\section{Changes undertaken to facilitate the learning of large networks}
\label{sec:net}
\noindent
When our interest is in learning a graphical model on a vertex set of
cardinality $p\gtrsim$20, it implies that such learning, if it is to be
undertaken according to the methodology described in the previous section,
will demand MCMC-based inference on the $\gtrsim$200 distinct off-diagonal elements of the
correlation matrix $\bSigma_{C}^{(S)}=[s_{ij}]$ (where
$S_{ij}$ represents correlation between the $X_i$ and $X_j$ variables in the
dataset); MCMC-based learning of more than about 200 parameters is difficult.
Again, for $p\gtrsim 500$, Cholesky decomposition of the $p\times
p$-dimensional inter-column correlation matrix $\bSigma_C^{(S)}$, (leading to
its inversion for example) is not easy (i.e. it is a challenge to achieve
numerical robustness as the matrix dimensionality exceeds about $500\times 500$). This
renders computation of the likelihood in Theorem~\ref{theo:marg} difficult, and
the numerical computation of the precision matrix $(\bSigma_C^{(S)})^{-1}=\bPsi=[\psi_{ij}]$ is also
difficult for $p\gtrsim 500$, where $\psi_{ij}$ is employed to compute the
partial correlation matrix $\bR$ according to Equation~\ref{eqn:6}.

\begin{remark}
When learning a network with $\gtrsim 500$ vertices as an SRGG
drawn in a probabilistic metric space, on vertex set $\bV$ with cardinality
$p$, with a cut-off on the affinity measure ($\equiv$ edge marginals) of $\tau$, we 
learn the SRGG given the correlation matrix $\bSigma_C^{(S)}$ than the partial
correlation matrix $\bR$ (of
the given data ${\bf D}_S$ that hosts $n$ standardised measurements of each of
the $p\gtrsim 500$ r.v.s), since it is hard to compute inverse
$(\bSigma_C^{(S)})^{-1}$ of the large $(\bSigma_C^{(S)})^{(p\times p)}$, to
thereby compute $\bR$. 

Thus, the network is learnt as the SRGG ${\cal G}_{m, \bSigma_C^{(S)}}(\bV,\tau)$.
\end{remark}

\begin{remark}
When learning a network with $\gtrsim 500$ vertices given data ${\bf D}_S$
that hosts $n$ standardised measurements of each of
$p\gtrsim 500$ r.v.s $X_1,\ldots,X_p$, we \
---eschew
MCMC-based inference on the large $(\bSigma_C^{(S)})^{(p\times p)}$
inter-column correlation matrix, and \\
---employ empirical estimate of $s_{ij}$
instead, where $\bSigma_C^{(S)}) = [s_{ij}]$ with $s_{ij}:=
\displaystyle{\frac{\sum\limits_{k=1}^n x_{ik} x_{jk}}{n} -
  \frac{\sum\limits_{k=1}^n x_{ik}}{n}\frac{\sum\limits_{k=1}^n
    x_{jk}}{n}}$. Here, $k$-th measured value of $X_i$ is $x_{ik}$, $i=1,\ldots,p;\:k=1,\ldots,n$. 

Hence in the notation of the network learnt as SRGG ${\cal G}_{m,
  \bSigma_C^{(S)}}(\bV,\tau)$, correlation matrix has no dependence on
any iteration index.
\end{remark}

\begin{remark}
  When learning the graphical model of a given dataset ${\bf D}_S$
  that hosts $n$ standardised measurements of each of $p\gtrsim 500$ r.v.s $X_1,\ldots,X_p$, we
  eschew MCMC-based inference on an SRGG in every iteration. The sought
  graphical model is learnt as the network ${\cal G}_{m,
  \bSigma_C^{(S)}}(\bV,\tau)$ which is itself an SRGG with connection function,
or the affinity measure between the $i$-th and $j$-th nodes, given by 
the marginal posterior of
$G_{ij}$, given correlation $S_{ij}=s_{ij}$ between $X_i$ and $X_j$, i.e. by: 
$$m(G_{ij}\vert S_{ij}).$$ 
\end{remark}
Indeed, as the MCMC-based inference in not relevant any more, there is only a
single value of the marginal posterior $m(G_{ij}\vert S_{ij})$ of the edge
parameter $G_{ij}$, between the $i$-th and $j$-th nodes, (given the
correlation $S_{ij}$).  So we do not require to define the connection function
in terms of (a sample estimate of) the expected value of the marginal.

\begin{remark}
Graphical model of dataset with inter-column correlation matrix $\bSigma_C$, on vertex set $\bV$, with cutoff probability $\tau$, 
is learnt as the network ${\cal G}_{m,\bSigma_C}(\bV,\tau)$ with one single 
identified connection function or affinity function $m(\cdot\vert\cdot)$. 

Thus, we learn a network as an SRGG without uncertainties.
\end{remark}

\subsection{Inter-network distance}
\noindent
However, in Definition~\ref{defn:1}, distance between graphical models learnt
of a pair of datasets, is defined as the Hellinger distance normalised by the
uncertainty in the learning of each graphical model. So in the absence of
uncertainty in learning the network as an SRGG, how can we define an
inter-network distance? In fact, the very discretised representation of the
Hellinger distance between the marginal posteriors of the two graphs, over the
MCMC iterations, (see
Equation~\ref{eqn:sq_hell}), stands challenged, when only one marginal
value of the SRGG is computed for each given dataset.

\begin{proposition}
For vertex set $\bV=\{X_1,\ldots,X_p\}$, 
distance $\Delta(\cdot,\cdot)$ between network ${\cal G}_{m, \bSigma_C^{(1)}}(\bV,\tau)$. given 
a dataset with inter-column correlation matrix $\bSigma_C^{(1)}$, and the
network ${\cal G}_{m, \bSigma_C^{(2)}}(\bV,\tau)$ learnt given 
dataset with inter-column correlation matrix $\bSigma_C^{(2)}$, is defined as 
the (discretised) Hellinger distance between the edge marginals of each
network, given the respective inter-node correlation structure, i.e. as
\begin{equation}
\Delta(\displaystyle{{{\cal G}}_{m, \bSigma_C^{(1)}}(\bV,\tau)},
\displaystyle{{{\cal G}}_{m, \bSigma_C^{(2)}}(\bV,\tau)}) 
:= D_H(u_{1},u_{2}), \nonumber\\
\label{eqn:delta13}
\end{equation}
where for the $q$-th dataset, ($q=1,2$), the marginal
posterior of the $ij$-th edge parameter $G_{ij}^{(q)}$ given the $ij$-th correlation
parameter $S_{ij}^{(q)}$ is $m(G_{ij}^{(q)}\vert
S_{ij}^{(q)})$, for $i> j; i,j\in\{1,2,\ldots,p\}$, s.t.
$$u_{q}^{(ij)}:=m(G_{ij}^{(q)}\vert S_{ij}^{(q)}),$$
which is employed to compute square of the (discretised version of the) Hellinger
distance between the two datasets as
\begin{equation}
D_H^2(u_{1},u_{2}) = \displaystyle{\frac{
\sum\limits_{i=1}^{p-1} \sum\limits_{j=i+1}^{p} \left(\sqrt{u_{1}^{(ij)}} - \sqrt{u_{2}^{(ij)}}\right)^2}{p(p-1)/2}},
\label{eqn:sq_hell2}
\end{equation}
\end{proposition}

The cut-off probability on this marginal posterior is $\tau$ in the 
network learnt as SRGG ${\cal G}_{m, \bSigma_C^{(S)}}(\bV,\tau)$. Depending on
the network at hand, we may decide on the value of $\tau$; for example, in the
human disease-disease network that we learn in Section~\ref{sec:disease}, we
produce the network using $\tau=0.1$.

\section{Implementation on real data}
\label{sec:real}
\noindent
\begin{sloppypar}
{In this section we discuss applications of our method to datasets on 12 
vino-chemical attributes of two samples of 1599 red and 4898 white wines, grown in the Minho region of Portugal
(referred to a ``vinho verde''); these data have been considered by
\ctn{CorCer09} and discussed in
\url{https://onlinecourses.science.psu.edu/stat857/node/223} (hereon PSU).
Each dataset consists of 12 columns that bear information on 
attributes that are assigned the following
names: ``fixed acidity'' ($X_1$), ``volatile acidity'' ($X_2$), ``citric
acid'' ($X_3$), ``residual sugar'' ($X_4$), ``chlorides'' ($X_5$), ``free
sulphur dioxide'' ($X_6$), ``total sulphur dioxide'' ($X_7$), ``density''
($X_8$), ``pH'' ($X_9$), ``sulphates'' ($X_{10}$), ``alcohol'' ($X_{11}$) and
``quality'' ($X_{12}$). Then the $n$-th row and $i$-th column of the data
matrix carries measured/assigned value of the $i$-th property of the $n$-th
wine in the sample, where $i=1,\ldots,12$ and $n=1,\ldots,n_{orig}=1599$ for
the red wine data ${\bf D}^{(red)}_{orig}$, while $n=1,\ldots,n_{orig}=4898$
for the white wine data ${\bf D}^{(white)}_{orig}$. We refer to the $i$-th
vinous property to be $X_i$. Then $X_i\in{\mathbb R}_{\geq 0}$ $\forall
i=1,\ldots,11$, while $X_{12}$ that denotes the perceived ``quality'' of the
wine is a categorical variable. Each wine in these samples was assessed by at
least three experts who graded the wine on a categorical scale of 0 to 10, in
increasing order of excellence. The resulting ``sensory score'' or value of
the ``quality'' parameter was a median of the expert assessments
\ctp{CorCer09}. We seek the graphical model given each of the wine data sets,
in which the relationship between any $X_i$ and $X_j$ is embodied, $i\neq j;\;
i,j=1,\ldots, 12$. Thus, we seek to find out how the different vino-chemical
attributes affect each other, as well as the quality of the wine, in the
data at hand. Here, $X_1,\ldots,X_{11}$ are real-valued, while $X_{12}$ is a categorical
variable, and our methodology allows for the learning of the graphical model
of a data set that in its raw state bears measurements of variables of
different types. In fact, we standardise our data, s.t. $X_i$ is standardised
to $Z_i$, $i=1,\ldots,p$, $p=12$. We work with only a subset data set,
(comprising only $n<n_{orig}$ rows of the available ${\bf
  D}^{(\cdot)}_{orig}$; $n=300$ typically). Thus, the data sets with $n$ rows,
containing $Z_i$ values, ($i=1,\ldots,p=12$), are $n\times p$-dimensional
matrices each; we refer to these data sets that we work with, as ${\bf
  D}^{(white)}_S$ and ${\bf D}^{(red)}_S$, respectively for the white and red
wines. Our aim is to learn the between-column correlation matrix
$\bSigma_S^{(m)}$ given data ${\bf D}^{(m)}_S$, and simultaneously
learn the graphical model of this data using MCMC-based inference within the methodology that we
discuss above, to then compute the inter-graph distance, and the inter-data
correlation thereafter; $m=white,\:red$.}
\end{sloppypar}

The motivation behind choosing these data sets are basically
three-fold. Firstly, we sought multivariate, rectangularly-shaped, real-life
data, that would admit graphical modelling of the correlations between the
different variables in the data. Also, we wanted to work with data, results
from -- at least a part of -- which exists in the literature. Comparison of
these published results, with our independent results then illustrates
strengths of our method. Thirdly, treating the red and white wine data as data
realised at different experimental conditions, we would want to address the
question of the distance between these data, and we propose to do this by
computing the distance between the graphical models of the two data
sets. Hence our choice of the popular Portuguese red and white wine data sets,
as the data that we implement to illustrate our method on. It is to be noted
that a rigorous vinaceous implications of the results, is outside the scope
and intent of this paper. However, we will make a comparison of our results
with the results of the analysis of white wine data that is reported in 
PSU,
though literature precludes analysis of the red wine data.

\subsection{Results given data ${\bf D}^{(white)}_S$}
\noindent
\begin{figure}[!t]
\centering
\includegraphics[width=10cm,height=8cm]{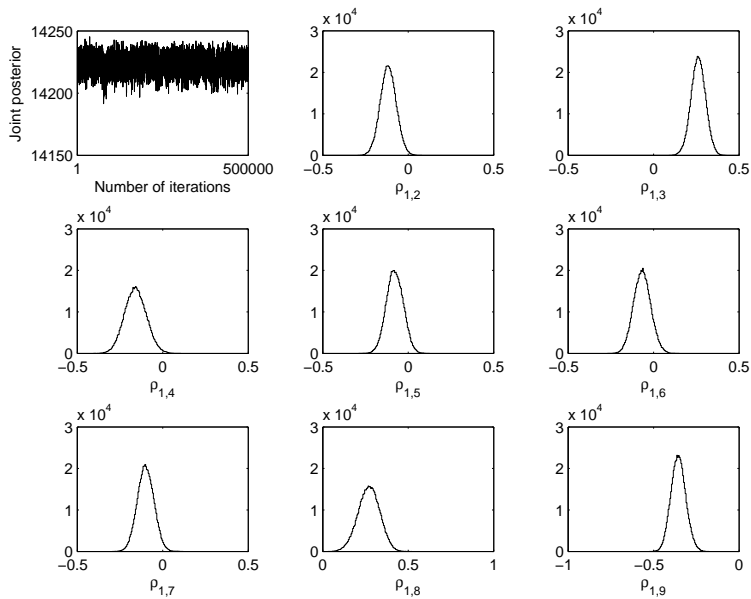}
\vspace{.0cm}
\caption{Figure to demonstrate convergence of the MCMC chain that we run with
  the white wine data. {\it Top left panel:} trace of the joint posterior
  probability density of the elements of the upper triangle of the
  between-columns correlation matrix of the standardised version of
  the real data ${\bf D}_S^{(white)}$ on Portuguese white wine samples
  \ctp{CorCer09}; this data has $n=300$ rows nd $p=12$ columns, and is
  constructed as a randomly sampled subset of the original data, the
  sample size of which is 4898. {\it All other panels:} histogram representations of marginal posterior probability densities of some of the partial correlation parameters computed using the correlation matrix learnt given data ${\bf D}_S^{(white)}$.
} 
\label{fig:white_corr}
\end{figure}

\begin{figure}[!t]
{
\hspace*{0in}
\includegraphics[width=10cm,height=8cm]{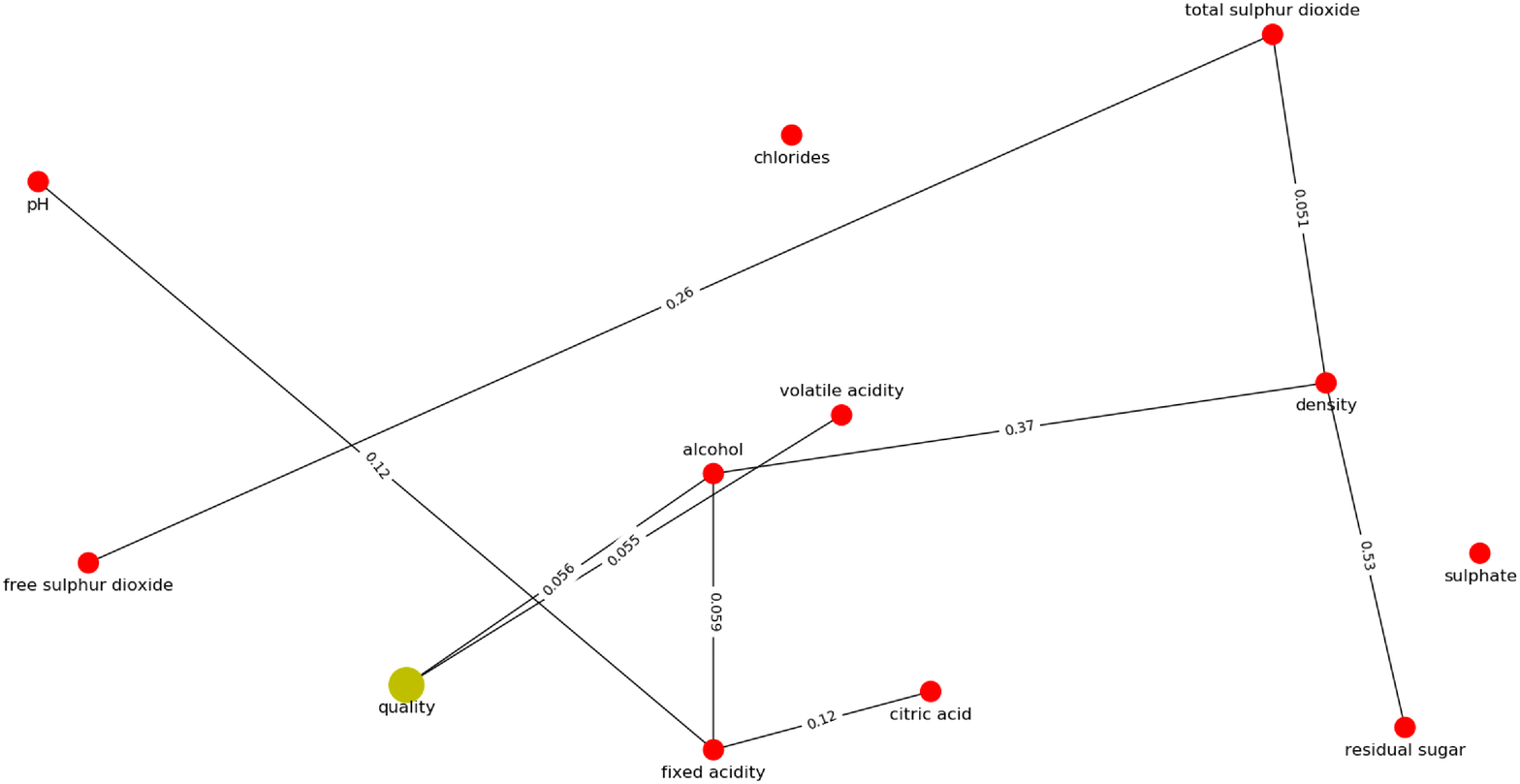}
}
\vspace*{-.05in}
\caption{Figure showing graphical model of ${\bf
    D}_S^{(white)}$, of the real data on Portuguese white wine samples
  \ctp{CorCer09}. The nodes have been placed randomly inside a box. Each of the first 11 columns of this data gives the measured
  value of each of 11 different vino-chemical properties of the wines in the
  sample -- marked as nodes in the graph above, by filled red (or grey in the
  printed version) circles, with the name of the property included in the
  vicinity of the respective node. The 12-th column in the data includes
  values of the assessed quality of a wine in the sample, (a node that we mark
  with a green circle in the electronic version; the bigger grey circle in a
  monochromatic version of the paper). The estimate of the probability for an edge to exist in
  the post-burnin sample of graphs generated in our MCMC-based inferential
  scheme, is marked against an existing edge, where edges with such
  probabilities that are $< 0.05$ are omitted from this graphical model, (see
  Section~\ref{sec:95}).}
\label{fig:white_graph}
\vspace*{-.1in}
\end{figure}

Figure~\ref{fig:white_corr} presents results that demonstrate convergence of the
MCMC chain run to learn the correlation structure and graphical model given
the white wine data ${\bf D}^{(white)}_S$. Its top left-hand panel displays trace of
the joint posterior probability density of all learnt inter-column correlation
($S_{ij}$) parameters of this data while marginal marginal posterior
probabilities of some of the partial correlation parameters, are presented as
histograms in the other panels. Having learnt the correlation structure, the
SRGG given the partial correlation matrix updated in an iteration, is then
learnt. Figure~9 in Supplementary Materials presents traces of the
of some of the edge ($G_{ij}$) and variance ($\upsilon_{ij}$) parameters of
this SRGG learning. Then at the end of the chain, using the learnt SRGGs, 
graphical model of this data is constructed; this is presented in
Figure~\ref{fig:white_graph}.

\subsubsection{Comparing against earlier work done with white wine data}
\label{sec:comparewhite}
\noindent   
The graphical model of the white wine data presented in
Fig~\ref{fig:white_graph}, is strongly corroborated by the simple empirical
correlations between pairs of different vino-chemical properties that is
noticed in the ``scatterplot of the
predictors'' included as part of the results of the ``Exploratory Data
Analysis'' reported in
\url{https://onlinecourses.science.psu.edu/stat857/node/224} on the white wine
data. These reported results use the
full white wine data set ${\bf D}^{(white)}_{orig}$, to construct a
matrix of scatterplots of $X_i$ against $X_j$; $i\neq j;\:
i,j=1,\ldots,11$. These empirical scatterplots
visually suggest stronger correlations between fixed acidity and pH;
residual sugar and density; free sulphur dioxide and total sulphur dioxide;
density and total sulphur dioxide; density and alcohol--than amongst other
pairs of variables. These are the very node pairs that we identify to have
edges (at probability in excess of 0.05) between them.

When we compare our learnt graphical model with the results of this
reported ``Exploratory Data Analysis'', we remind ourselves that
partial correlation (that drives the probability of the edge between
the $i$-th and $j$-th nodes), is often smaller than the correlation
between the $i$-th and $j$-th variables, computed before the effect of
a third variable has been removed \ctp{sheskin}. If this is the
case, then an edge between nodes $i$ and $j$ in the learnt graphical
model, is indicative of a high correlation between the $i$-th and
$j$-th variables in the data. However, in the presence of a suppressor
variable (that may share a high correlation with the $i$-th variable,
but low correlation with the $j$-th), the absolute value of the
partial correlation parameter can be enhanced to exceed that of the
correlation parameter. In such a situation, the edge between the nodes
$i$ and $j$ in our learnt graphical model may show up (within our
defined 95$\%$ HPD credible region on edge probabilities, i.e. at
probability higher than 0.05), though the empirical correlation
between these variables is computed as low \ctp{sheskin}. So, to
summarise, if the empirical correlation between two variables reported
for a data set is high, our learnt graphical model should include an
edge between the two nodes. But the presence of an edge between pair
of nodes is not necessarily an indication of high empirical
correlation between a pair of variables--as in cases where suppressor
variables are involved. Guessing the effect of such suppressor
variables via an examination of the scatterplots is difficult in this
multivariate situation. Lastly, it is appreciated that empirical trends are
only indicators as to the matrix-Normal density-based model of the learnt correlation structure (and the graphical model learnt thereby) given the data at hand. 

Effect on the ``quality''
variable in the ``Exploratory Data Analysis'' reported in PSU site, using the
white wine data, is examined via a linear regression analysis of the predictors $X_1,\ldots,X_{11}$ on the
response variable ``quality'', which suggests the variables
alcohol and volatile acidity to have maximal effect on quality. 
Indeed, this is corroborated in our learning
of the graphical model that manifests edges between the nodes corresponding to
variables: alcohol-quality, and volatile acidity-quality.

\begin{figure}[!t]
{
\hspace*{0in}
\includegraphics[width=10cm,height=8cm]{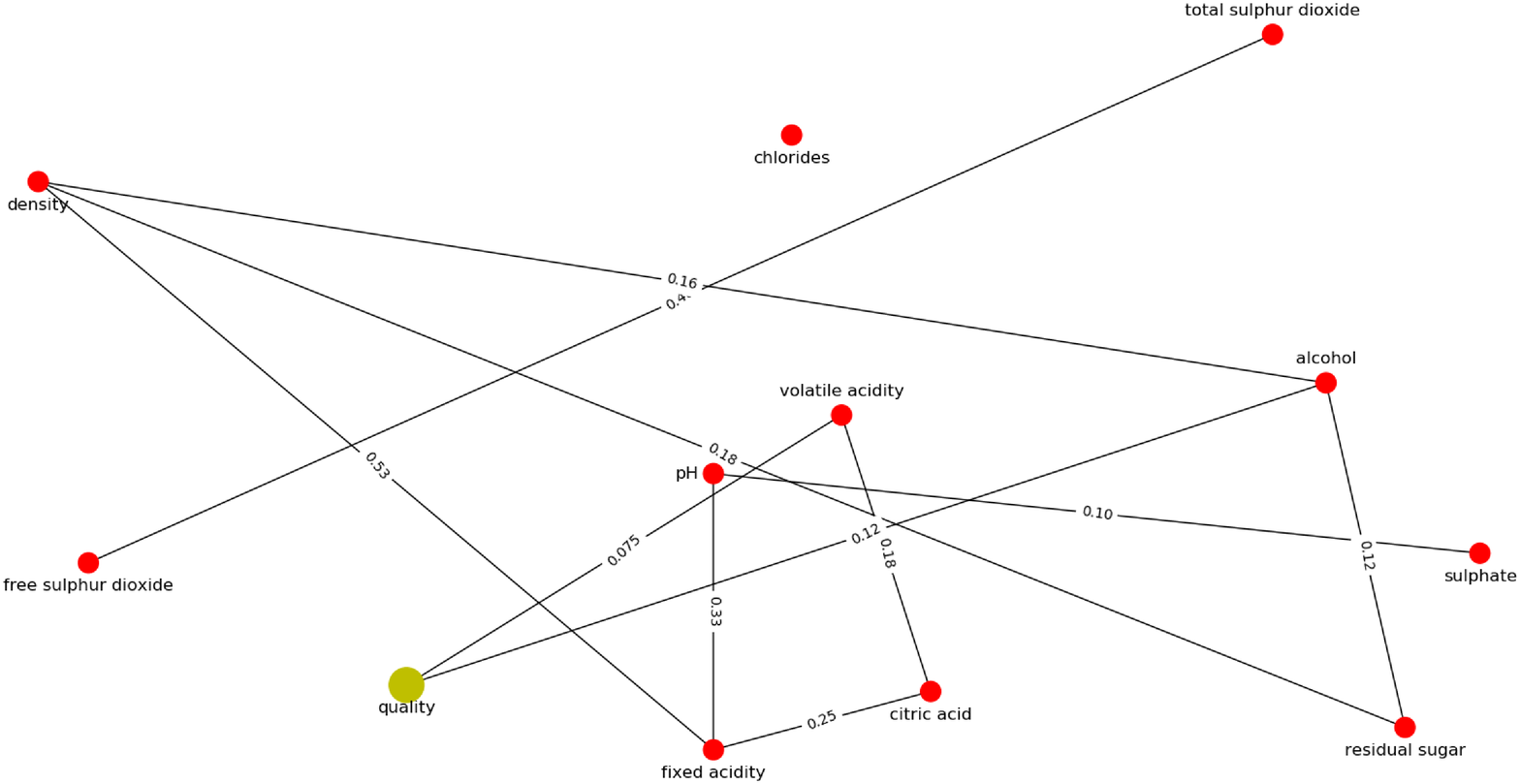}
}
\vspace*{-.05in}
\caption{Graphical model of standardised version ${\bf D}_S^{(red)}$ of the
  real data on Portuguese red wine samples
  \ctp{CorCer09}. Figure is similar to Figure~\ref{fig:white_graph}, even in
  the random placement of the nodes, except that this is the graphical model learnt for the red wine data.}
\label{fig:red_graph}
\vspace*{-.2in}
\end{figure}

\subsection{Results given data ${\bf D}^{(red)}_S$}
\label{sec:red}
\noindent
The ${\bf D}^{(red)}_S$ data is the standardised version of a subset of the
original red wine data set ${\bf D}^{(red)}_{orig}$. ${\bf D}^{(red)}_S$
comprises $n=300$ rows and $p=12$. 
The marginal posterior of some of the partial correlation parameters
$\rho_{ij}$ computed using the elements of the correlation matrix
$\bSigma^{(red)}_S$ (of data ${\bf D}^{(red)}_S$) that is updated in the first
block of Metropolis-with-2-block-update, are presented in Figure~10 of the
Supplementary Material.
Figure~11 of the Supplementary Material presents the trace of the joint
posterior probability of the edge ($G_{ij}$) parameters and the variance
($\upsilon_{ij}$) parameters learnt given data ${\bf
  D}^{(red)}_S$. The inferred graphical model of the
red wine data is included in Figure~\ref{fig:red_graph}.
\subsubsection{Comparing against empirical work done with red wine data}
\noindent   
To the best of our knowledge, analysis of the red wine data has not been
reported in the literature. In lieu of that, we undertake construction of a
matrix of scatterplots of $X_j$ against $X_i$ from the red wine data. This is
shown in Figure~12 of the Supplementary Material, for $i=1,\ldots,11$. These
scatterplots visually indicate moderate correlations between the following
pairs of variables: fixed acidity-citric acid, fixed acidity-density, fixed
acidity-pH, volatile acidity-citric acid, free sulphur dioxide-total sulphur
dioxide, density-alcohol. All these variables share an edge at probability
$\geq 0.05$ in our learnt graphical model of data ${\bf D}_S^{(red)}$ (see
Figure~\ref{fig:red_graph}). We note that all moderately correlated variable
pairs, as represented in these scatterplots, are joined by edges in our learnt
graphical model of the red wine data -- as is to be expected if the learning
of the graphical model is correct. Such pairs include fixed acidity-citric
acid, fixed acidity-density, fixed acidity-pH, volatile acidity-citric acid,
free sulphur dioxide-total sulphur dioxide, density-alcohol. However, an edge
may exist between a pair of variables even when the apparent empirical
correlation between these variables is low, owing to the effect of other
variables (discussed in Section~\ref{sec:comparewhite}).

Noticing such edges from the residual-sugar variable, we undertake a
regression analysis (ordinary least squares) with residual-sugar regressed
against the other remaining 10 vino-chemical variables. The MATLAB output of
that analysis carried out using the red wine data, is included in Figure~13 of
the Supplementary Material. The analysis indicates that the covariates with
maximal (near-equal) effect on residual-sugar, are density and alcohol;
indeed, in our learnt graphical model of the red wine data
(Figure~\ref{fig:red_graph}), residual-sugar is noted to enjoy an edge with
both density ($Z_7$) and alcohol ($Z_{10}$)

We also undertook a separate ordinary least squares analysis with the response
variable quality, regressed against the vino-chemical variables as the
covariates. The MATLAB output of this regression analysis in in Figure~14 of
the Supplementary Section. We notice that the strongest (and nearly-equal) effect on quality is from the variables volatile-acidity and alcohol--the very two variables that share an edge at probability $\geq 0.05$ with quality, in our learnt graphical model of the red wine data.

\section{Metric measuring distance between posterior probability densities of graphs given white and red wine datasets}
\label{sec:real_hell}
\noindent
We seek the distance $\delta(\cdot,\cdot)$ that we defined in
Definition~\ref{defn:1}, between the learnt
red and white wine graphs, using the method delineated in
Section~\ref{sec:hell}. For this, we first compute 
the normalisation\\ $S:={\max}\{(\ln(p_{red}^{(0)}),
\ln(p_{red}^{(1)}),\ldots,\ln(p_{red}^{(N_{iter})}),
\ln(p_{white}^{(0)}),\ldots,\ln(p_{red}^{(N_{iter})})\}$, which for the red
and white wine datasets yields
$s=\ln(p_{red}^{(1474)})\approx 142.7687$. We then use $\exp(\ln(u_m^{(t)})/s)$ in Equation~\ref{eqn:sq_hell}; $m=white, red$. 
Then scaling the log posterior given either data set, at
any iteration, by the global scale value of $s$=142.7687 approximately, we get
$D_H(u_{white},u_{red})\approx 0.1153$, so that the logarithm of this
value of the Hellinger distance between the 2 learnt graphical models is $\ln(0.1153)\approx -2.1602$.
Similarly, using the same scale, the Bhattacharyya distance is
$D_B(p_{white},p_{red})\approx -1.7623$, where we recall that this
measure is a logarithm of the distance.

\begin{figure}[!hbt]
{
\hspace*{0in}
\includegraphics[width=12cm]{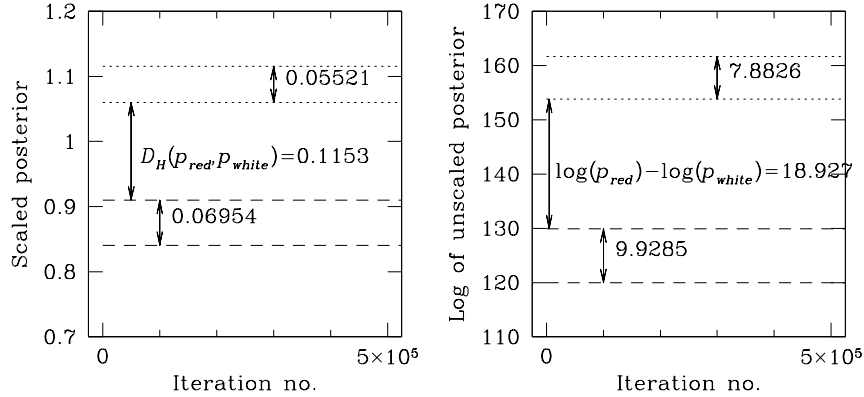}
}
\vspace*{-.0in}
\caption{{\it Left}: minimum and maximum values of the scaled posterior
  probability density of the SRGG sampled in an iteration in the MCMC chain
  run with the red wine data, plotted in dotted lines against the number of
  the iteration. The difference between these values is depicted within the
  band delineated by these lines. The broken lines show the same for the
  results obtained from the MCMC chain run using the white wine data. The
  value of the Hellinger distance computed using the
  scaled posterior probabilities of the graphical models given the two wine
  data sets, is also marked, as about 0.1153. All log posterior values are
  scaled by a chosen global scale (of about 143), and exponentiated (as discussed in the
  text). {\it Right}: similar to the left panel, except that here, the ratio of
  the logarithm of the unscaled posteriors is used; the value of the log odds
  between the posteriors of the red and white wine data sets is marked to be
  about 18.927.}
\label{fig:dist}
\end{figure}
For this $s$ and the red wine data, we compute the uncertainty inherent in
graphical model of the red-wine data as $D_{max,s}(red)$,
between the graph that occurs at maximal posterior and that at the minimal
posterior (Equation~\ref{eqn:graph_new}). Similarly, we compute
$D_{max,s}(white)$. We then compute ratio of the Hellinger distance between the
graphical models learnt given the red and white-wine data, to the uncertainty
inherent in each learnt model, and compare
${D_H(p_{white},p_{red})}/D_{max,s}(red)$, with
${D_H(p_{white},p_{red})}/D_{max,s}(white)$. This comparison is
depicted in the left panel of Figure~\ref{fig:dist} that shows that the
difference $D_{max,s}(white)$ between the scaled posterior of graphs given the
white wine data is about 0.0694 while $D_{max,s}(red)$ given the red wine data
is about 0.05521, These values are compared to the Hellinger distance (between
scaled posteriors) of about 0.1153, between graphs given the red and white
wine data.  Thus, $D_H(u_{red},u_{white})$ is about 1.66$D_{max,s}(white)$ and
about 2.1$D_{max,s}(red)$. 
Thus, our inter-graph distance metric, between the graphical models learnt
given the two data sets is $$\delta(white, red)\approx 0.44$$. 
Then intuitively speaking, this inter-graph distance between the graphical models given the red and white wine
datasets, may suggest independence of the data sets. 

Again, using the correlation model suggested in Proposition~\ref{prop:2}, the
absolute value of the correlation between the 12-dimensional vino-chemical
vector-valued measurable for the red wine data and that for the white wine
data, is $$\vert corr(white, red)\vert := \exp[-\delta(white, red)]\approx
0.1030,$$ which is a low correlation, indicating that the two graphical models
learnt given the real red and white wine Portuguese datasets, are not sampled
from the same {\it{pdf}}.

Compared to these, the sample mean of the log odds of the posterior of the
graphs generated in the post-burnin iterations, given the two data is 18.9273,
which is about 1.9 times the maximal difference between the log posterior
values of graphs achieved in the MCMC run with the white wine data, and about
2.4 times that for the red wine data (see Figure~\ref{fig:dist}). Again,
this suggests that the log odds as a measure of divergence between the
graphical models given these two wine data sets, is significantly higher than
the uncertainty internal to the results for each data.

This clarifies how our pursuit of uncertainties in learnt graphical models,
and inter-graph distance, share an integrated umbrage of purpose, where the
former leads to the latter.

\section{Learning the human disease-symptom network}
\label{sec:disease}
\noindent
Our methodology for learning the graphical model, can be
implemented even for a highly multivariate data that generates a
graph with a very large number of nodes. In this section, we discuss such a
graph (with $\gtrsim$8000 nodes) that describes the correlation structure of
the human disease-symptom network. 
\begin{figure}[!hbt]
\hspace*{-1.0cm}{
\includegraphics[width=15cm,height=12cm]{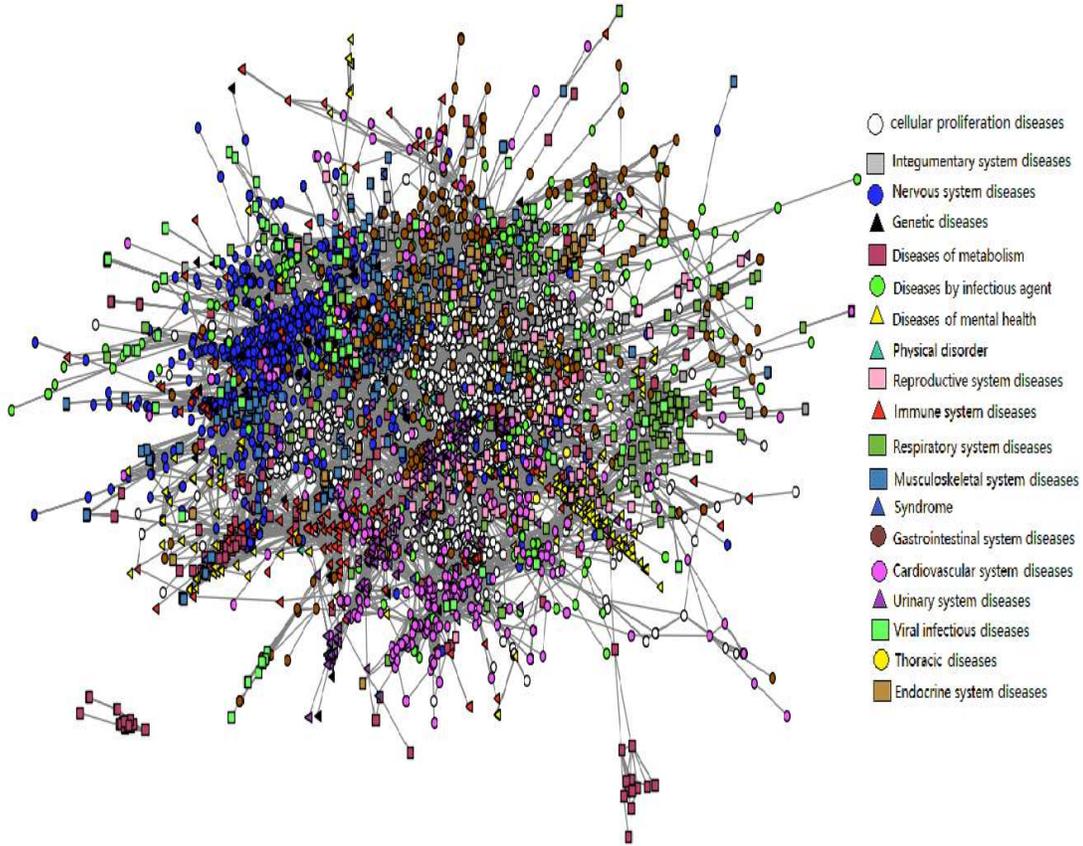}}
\vspace{-.5cm}
\caption{SRGG ${\cal G}_{m, {\bf D}_{DPh}}(\bV,0.1)$, representing the human
  disease-phenotype network that we learn using the disease-disease partial
  correlation obtained using the computed Spearman rank correlation between
  the rank vectors of a list of phenotypes, where the phenotype ranking
  reflects semantic relevance of a phenotype to the disease in question
  (quantified by HSG as the NPMI parameter in the ${\bf D}_{DPh}$ dataset). In
  our learnt SRGG, $\tau=0.1$, i.e. only edges (between the $i$-th and $j$th
  diseases bearing a Spearman rank correlation of ${\mathpzc{s}}_{ij}^{(rank)}$), with marginal posterior
  $m(G_{ij}\vert {\mathpzc{S}}_{ij}^{(rank)})\geq 0.9$ are included in this graph. Here cardinality
  of vertex set $\bV$ is 8676, but all nodes with no edges are discarded from
  this visualised graph, resulting in 6052 diseases (nodes) and 145210 edges remaining
  that are shown this figure. Diseases identified by HSG, to belong to one of
  the 19 given disease class, are presented above in the same colour; the
  colour key identifying these classes, is attached. To draw the graph, we
  used a Python-based code that implements the Fruchterman-Reingold
  force-directed algorithm.}
\label{fig:disgraph}
\end{figure}

\ctn{hsg} (HSG hereon) learn this network by considering the similarity
parameter for each pair of diseases that are elements of an identified set of
diseases in the Human Disease Ontology (DO), that contains information about
rare and common diseases, and spans heritable, developmental, infectious and
environmental diseases. Here, the ``similarity parameter'' between one disease
and another, is computed using the ranked vectors of ``normalised pointwise
mutual information'' (NMPI) parameters for the two diseases, where the NMPI
parameter describes the relevance of a symptom (or rather, a phenotype), to
the disease in question. HSG define the NMPI parameter semantically, as the
normalised number of co-occurrences of a given phenotype and a disease in the
titles and abstracts of 5 million articles in Medline. To do this, they make
use of the Aber-OWL: Pubmed infrastructure that performs such semantical
mining of the Medline abstracts and titles. The disease-disease pairwise
semantic similarity parameters -- computed using the degree of overlap in the
relevance ranks of phenotypes associated with each disease -- result in a
similarity matrix, which HSG turn into a disease–disease network based on
phenotypes. To do this, they only choose from the top-ranking 0.5$\%$ of
disease–disease similarity values. Phenotypes associated with diseases, and
corresponding scoring functions (such as the NPMI), exist in the file
``doid2hpo-fulltext.txt.gz'' at
\url{http://aber-owl.net/aber-owl/diseasephenotypes}. In fact,
this file
contains information about $N_{dis}$ diseases, and
the semantic relevance of each of the $N_{pheno}$ phenotypes to each disease,
as quantified by NPMI parameter values, in addition to other scores such as
$t$-scores and $z$-scores. In this file, $N_{dis}$ is 8676 and $N_{pheno}$ is
19323. In the phenotypic similarity network between diseases that HSG report,
diseases are the nodes, and the edge between two nodes exists in this
undirected graph, if the similarity between the nodes (diseases) is in the
highest-ranking 0.5$\%$ of the 38,688,400 similarity values. They remove all
self-loops and nodes with a degree of 0. Their network is
presented in \url{http://aber-owl.net/aber-owl/diseasephenotypes/network/}.
The network analysis was performed using standard softwares and they identify
multiple clusters in their network, with agglomerates of some clusters (of
diseases), found to correspond to known disease-classes. The ``Group
Selector'' function on their visualisation kit, allows for the identification
of 19 such clusters in their disease-disease network, with each cluster
corresponding to a disease-class. 
Total number of nodes over their
identified 19 clusters, is 5059. The number of edges in their network is
reported to be 65,795; average node degree$\approx$26.2. 

We discuss
detailed comparison of our results to HSG's in the following subsection,
including comparison of HSG's and our recovery of
the relative number of nodes i.e. diseases, in each of the 19
disease classes that HSG classify their reported network into, and our
computed ratios of the averaged intra-class to inter-class variance for     
each of the 19 classes, compared to the ROC Area Under Curve       
values reported by HSG for each class.

\begin{figure}[!hbt]
\centering
\includegraphics[width=13cm,height=6cm]{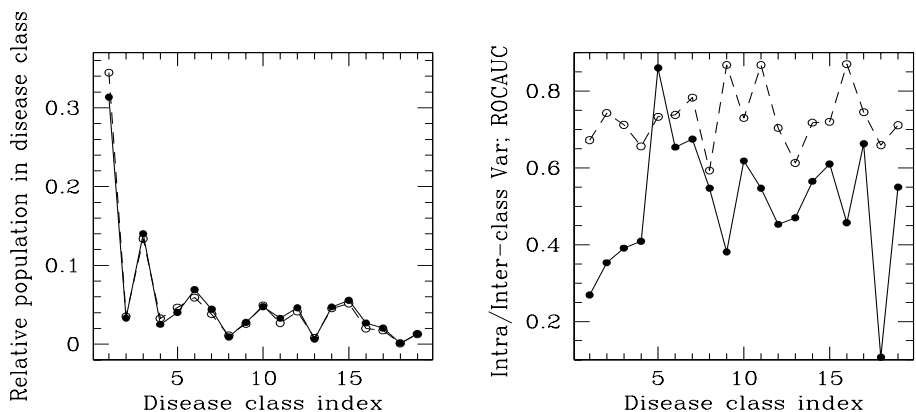}
\vspace{-.5cm}
\caption{{\it Left:} comparison of the
  relative number of nodes (diseases) that we recover in each of the 19
  disease classes that HSG classify their reported network to
  be classified into, with the relative class-membership reported by HSG. Our
  results are shown as filled circles joined by solid lines. In open circles
  threaded by broken lines, we overplot the relative
  number of diseases in each of the 19 classes, as reported by HSG. Similarity of the relative
  populations in the different disease classes, indicate that our learnt
  clustering distribution is similar to that obtained by HSG. {\it Right:}
our computed ratios of the averaged intra-class to inter-class variance for
each of the 19 classes, shown in filled circles; the ROC Area Under Curve
values reported by HSG for each class, is overplotted as open
circles joined by broken lines. The disease class
indices, from assigned values of 1 to 19, are the following respectively:
cellular proliferation diseases, integumentary diseases, diseases of the
nervous system, genetic diseases, diseases of metabolism, diseases by
infectious agents, diseases of mental health, physical disorders, diseases of
the reproductive system, of the immune system, of the respiratory system, of
the muscleoskeletal system, syndromes, gastrointestinal diseases,
cardiovascular diseases, urinary diseases, viral infections, thoracic
diseases, diseases of the endocrine system.}
\label{fig:compgraph}
\end{figure}

HSG's network then manifests a similarity-structure that is computed using 
available NPMI parameter values. Our interest is in learning the
disease-disease network as an SRGG, with each edge of such a graphical model
learnt to exist at a learnt probability $\tau$. We perform such learning using the
NPMI semantic-relevance data that is made available for each of the $N_{dis}$
number of diseases, by HSG; so $N_{dis}$ is the cardinality of the vertex set
$\bV$ of our sought SRGG. We refer to this human
disease-phenotype data as ${\bf D}_{DPh}$. Using ${\bf D}_{DPh}$, we first
compute the correlation $S_{ij}$ between the $i$-th and $j$-th diseases in $\bV$, for each of
which, information on the ranked (semantic) relevance of each of the
$N_{pheno}$ phenotypes exist, in this given dataset. Upon computation of
pairwise correlations, the SRGG for the data ${\bf D}_{DPh}$
is learnt.

We compute the correlation between the $i$-th and $j$-th
diseases in the ${\bf D}_{DPh}$ data, ($i,j=1,\ldots,N_{dis}$, $i\neq j$), in
the following way. We rank the NPMI parameter values for the $i$-th
disease and each of the $N_{pheno}$ phenotypes, with the phenotype of the 
highest semantic relevance to the $i$-th disease assigned a rank 1. Let the 
rank vector of phenotypes, by semantic relevance to the $i$-th disease take the value ${\bf{{\mathpzc{r}}_i}}$
and similarly, the rank vector of phenotypes relevant to the $j$-th disease
is ${\bf{{\mathpzc{r}}_j}}$. We compute the Spearman rank correlation ${\mathpzc{s}}_{ij}^{(rank)}$,
of vectors ${\bf{{\mathpzc{r}}_i}}$ and ${\bf{{\mathpzc{r}}_j}}$. Then we
compute this rank correlation ${\mathpzc{s}}_{ij}^{(rank)}$ 
$\forall\:i,j=1,\ldots,N_{dis};\:i\neq j$, between
the $i$-th and $j$-th nodes. The Spearman rank correlation is preferred to the correlation between the vectors of normalised NPMI values, since we intend to correlate the $i$-th disease with the $j$-th disease, depending on how relevant a given list of phenotypes is, to each disease, i.e. depending on the ranked relevance of the phenotypes. 
We learn the network given this correlation structure, 
that is itself computed using data ${\bf D}_{DPh}$ (see
Section~\ref{sec:net} on learning large networks).

\begin{definition}
  Our visualised SRGG in Figure~\ref{fig:disgraph} is a sub-graph of the full graph
  ${\cal G}_{m, {\bf D}_{DPh}}(\bV, 0.1)$ where $\bV$ has cardinality
  $N_{dis}$, and 
the inter-column
  correlation matrix of data ${\bf D}_{DPh}$ is $\bSigma_C^{(S)}=[{\mathpzc{s}}_{ij}^{(rank)}], \:i\neq j, \;
  i,j=1,\ldots,N_{dis}$, such that this visualised graph is defined to
  consist only of nodes with non-zero degree. This visualised graph has 6052 number of nodes
  (diseases) and 145210 edges, so that the average node degree is
  about 24. This undirected SRGG represents our learning of the human disease
  phenotype graph (displayed in Figure~\ref{fig:disgraph}).
\end{definition}  

\subsection{Comparing our results to the earlier work done on the human disease-symptom network}
\noindent
The ``Group
Selector'' function on the visualisation kit that HSG use, allows for the identification
of 19 such clusters in their disease-disease network, with each cluster
corresponding to a disease-class. 
This function also allows identification of
the number of diseases (i.e. nodes) in each disease-class (see left panel of
Figure~\ref{fig:compgraph}). The right
panel of Figure~\ref{fig:compgraph} displays the ratio of intra-class variance 
to the inter-class variance of each disease-class; the value of the area under
the Receiver Operating Characteristic curve (ROCAUC) for each cluster is
overplotted, where the ROCAUC value for the $i$-th cluster can be
interpreted as probability that a randomly chosen node is ranked as more
likely to be in the $i$-th class than in the $j$-th class; $i\neq j;\:
i,j=1,\ldots,19$ \ctp{caspian}.

\section{Conclusion}
\noindent
In this work, we present a methodology for Bayesianly learning a Soft Random
Geometric Graph that is drawn in a probabilistic metric space, allowing for
the connection function of this SRGG to be equated to the marginal posterior
of the graph edge parameter, given the correlation between the points that
this edge connects, with the threshold radius on this SRGG to be rendered a
probability, s.t. only edges with marginals that exceed such a threshold
(probability $\tau\in[0,1]$) are included in the graph. We demonstrate the
SRGG as generated from a point process that we identify as a non-homogenous
Poisson process, with intensity that varies with the node. 

In fact, correlation between each pair of nodes is learnt as well, and the
SRGG updated at each update of the correlation matrix, within each iteration
of the iterative inference scheme that we employ; (to be precise, the
MCMC-based inference). Here, each of the $p$ nodes of the graph is a variable
$X_i$ -- $n$ measurements of each of which -- comprises the dataset, the
standardised version of which we learn the graphical model and the correlation
matrix of. To be precise, the $i$-th column of the dataset contains the $n$
measurements of the r.v. $X_i$, standardised by its sample mean and standard
deviation; $i=1,\ldots,p$. The vertex set of the sought SRGG is then
$\bV=\{X_1,\ldots,X_p\}$. The continuous-valued generative process of the
inter-column correlation matrix, is identified after we achieve closed-form
marginalisation of the joint likelihood of the inter-column and inter-row
correlation matrices given the dataset, over all possible inter-row
correlations. The resulting process underlying the inter-column corelation, is
then compounded with the non-homogeneous Poisson point process, to generate
the SRGG.
The graphical model of the data is identified with 95$\%$ HPDs, on vertex set
$\bV$, to be the graph with edges, the expected marginals of which exceed
5$\%$, where a sample estimate of the expected marginal of an edge is provided
by its relative frequency from across the sample of SRGGs that are realised
across the iterations of the undertaken inference. When learning a large
network, such an iterative inferential scheme is prohibitively expensive
though. So then we learn the inter-column correlation of the given dataset
empirically, and employ it to learn the SRGG that represents this network.

Our Bayesian learning approach allows for acknowledgement of measurement
errors of any observable. The effect of ignoring such existent
measurement errors, on the graphical model, is demonstrated using a simple,
low-dimensional simulated dataset (see Section~4.2 of the Supplementary
Material; compare Figures~4 and 6 of the Supplementary Materials). Even in such a low-dimensional example, the difference made to the
inferred graph of the given data, by the inclusion of measurement errors, is
clear.

Ultimately we aim at
computing the distance between a pair of such learnt graphical models, of
respective datasets. To compute this inter-graph
distance, we advance a new metric that is given by the difference between the
Hellinger distance between the posterior probabilities of the
graphs, normalised by the uncertainty in one of the learnt graphs, and the
Hellinger distance normalised by the uncertainty in the other learnt graph. 

This novel, eventual computation of the inter-graph
distance is important in the sense that it informs on the
correlation structure of a dataset that is higher-dimensional than being
rectangularly-shaped, such as a cuboidally-shaped dataset that comprises slices of
rectangularly-shaped data slices. Then, the distance between
the graphical models of a pair of such slices of data, will inform on
the correlation between such slices of data. Such information is easily
calculated under the approach discussed herein, even when the datasets are
differently sized, and highly multivariate. An example 
could be a large network observed on a sample of size $n_1$ before an
intervention/treatment, and after implementation of such intervention,
when a smaller sample (of size $n_2$; $n_2\neq n_1$) is investigated. We
illustrate this on computing distance between
the learnt vino-chemical graphical models of
Portuguese red and white wine samples. 

Our learning of large networks is illustrated by the human disease-phenotype network (with $\geq$8,000 nodes).
In this application, learning the inter-node correlation was cast into a
semantic exercise in which we learnt the Spearman rank correlation between 
vectors of associated phenotypes, where any phenotype vector is ranked in order of
relevance to the disease in question. Other situations also admit such possibilities, for example, the
product-to-product, or service-to-service correlation in terms of associated
emotion, (or some other response parameter), can be semantically gleaned from
the corpus of customer reviews uploaded to a chosen internet facility, and the
same used to learn the network of products/services. Importantly, this method
of probabilistic learning of small to large networks, is useful for the
construction of networks that evolve with time, i.e. of dynamic networks.

\renewcommand\baselinestretch{1.}
\small
\bibliographystyle{ECA_jasa}

\end{document}